\documentclass[11pt]{article}

\PassOptionsToPackage{numbers}{natbib}
\usepackage[utf8]{inputenc} %
\usepackage{natbib}
\usepackage[letterpaper, portrait, margin=1in]{geometry}
\usepackage[T1]{fontenc}    %
\usepackage[bookmarks, colorlinks=true, citecolor=Violet,linkcolor=Mahogany,urlcolor=blue]{hyperref}
\usepackage{url}            %
\usepackage{booktabs}       %
\usepackage{amsfonts}       %
\usepackage{nicefrac}       %
\usepackage{microtype}      %

\usepackage[usenames,dvipsnames,dvinames]{xcolor}
\usepackage{multirow}
\usepackage{microtype}
\usepackage{enumitem}
\usepackage{amsmath,amsthm}
\usepackage[mathscr]{eucal}
\usepackage{amssymb}
\usepackage{footnote}
\usepackage{graphicx} %
\usepackage{algorithmic}
\usepackage{ifthen}
\usepackage[T1]{fontenc}
\usepackage{lmodern}
\providecommand{\dodraft}{true}
\providecommand{\doarxiv}{false}
\newboolean{isdraft}
\setboolean{isdraft}{\dodraft} 
\ifthenelse{\boolean{isdraft}}{

  \newcommand{\jeff}[2]{{\color{red}{#1 $\to$ {\bf Jeff}}: #2}}
  \newcommand{\wenruo}[2]{{\color{OliveGreen}{#1 $\to$ {\bf Wenruo}}: #2}}
  
  \newcommand{\everyone}[2]{{\color{Bittersweet}{#1 $\to$ {\bf Everyone}}: #2}}

} {

  \newcommand{\jeff}[2]{}

  \newcommand{\wenruo}[2]{}
  \newcommand{\everyone}[2]{}

}

\usepackage[vlined,ruled]{algorithm2e}
\usepackage{wrapfig}
\usepackage{float}
\usepackage{caption}

\usepackage{multicol}

\usepackage{bbm}
\usepackage{mathtools}

\DeclareMathOperator*{\argmax}{argmax}

\newtheorem{theorem}{Theorem}[section]
\newtheorem{corollary}{Corollary}[theorem]
\newtheorem{definition}[theorem]{Definition}

\newtheorem{lemma}[theorem]{Lemma}
\newtheorem{proposition}[theorem]{Proposition}

\newcommand{\cnstr}{\ensuremath{\Upsilon}}

\newcommand{\curv}{\ensuremath{\kappa}}
\newcommand{\curvf}{\ensuremath{\kappa_{f}}}
\newcommand{\curvg}{\ensuremath{\kappa^g}}

\newboolean{isarxiv}
\setboolean{isarxiv}{\doarxiv} 
\ifthenelse{\boolean{isarxiv}}{%
  \newcommand{\arxiv}[1]{#1}
  \newcommand{\notarxiv}[1]{}
}{
  \newcommand{\arxiv}[1]{}
  \newcommand{\notarxiv}[1]{#1}
}

\newcommand{\arxivalt}[2]{\ifthenelse{\boolean{isarxiv}}{#1}{#2}}
\newcommand{\arxivaltr}[2]{\ifthenelse{\boolean{isarxiv}}{#2}{#1}}

\newcommand{\set}[1]{\left\{#1\right\}}
\usepackage{subcaption}

\usepackage{thmtools,thm-restate}
\title{Greed is Still Good: Maximizing Monotone Submodular+Supermodular Functions}

\author{
  Wenruo Bai\\
  Department of Electrical Engineering\\
  University of Washington\\
  Seattle, WA 98195 \\
  \texttt{wrbai@uw.edu} \\
  \and
  Jeffrey A. Bilmes\\
  Department of Electrical Engineering\\
  University of Washington\\
  Seattle, WA 98195 \\
  \texttt{bilmes@uw.edu} \\
}

\newlength{\spacebeforeproofref}
\begin{document}

\maketitle\maketitle
\thispagestyle{empty}

\begin{abstract}
  We analyze the performance of the greedy algorithm, and also a
  discrete semi-gradient based algorithm, for maximizing the sum of a
  suBmodular and suPermodular (BP) function (both of which are
  non-negative monotone non-decreasing) under two types of
  constraints, either a cardinality constraint or $p\geq 1$ matroid
  independence constraints.  These problems occur naturally in several
  real-world applications in data science, machine learning, and
  artificial intelligence.  The problems are ordinarily inapproximable
  to any factor (as we show).  Using the curvature $\curv_f$ of the
  submodular term, and introducing $\curv^g$ for the supermodular term
  (a natural dual curvature for supermodular functions), however, both
  of which are computable in linear time, we show that BP maximization
  can be efficiently approximated by both the greedy and the
  semi-gradient based algorithm.  The algorithms yield multiplicative
  guarantees of
  $\frac{1}{\curv_f}\left[1-e^{-(1-\curv^g)\curv_f}\right]$ and
  $\frac{1-\curv^g}{(1-\curv^g)\curv_f + p}$ for the two types of
  constraints respectively. For pure monotone supermodular constrained
  maximization, these yield $1-\curvg$ and $(1-\curvg)/p$ for the two
  types of constraints respectively.  We also analyze the hardness of
  BP maximization and show that our guarantees match hardness by a
  constant factor and by $O(\ln(p))$ respectively. Computational
  experiments are also provided supporting our analysis.
\end{abstract}

\newpage
\clearpage
\setcounter{page}{1}
\section{Introduction}

The \emph{Greedy
  algorithm}~\cite{bednorz2008advances,cormen2009introduction} is a
technique in combinatorial optimization that makes a
locally optimal choice at each stage in the hope of finding a good
global solution. It is one of the simplest, most widely applied, and
most successful algorithms in
practice~\cite{kempner2008quasi,zhang2000greedy,karp2000gpsr,ruiz2007simple,wolsey1982analysis}. Due
to its simplicity, and low time and memory complexities, it is used
empirically even when no guarantees are known to exist although, being
inherently myopic, the greedy algorithm's final solution can be
arbitrarily far from the optimum solution~\cite{BANGJENSEN2004121}.

On the other hand, there are results going back many years showing where the
greedy algorithm is, or almost is, optimal, including Huffman
coding~\cite{huffman1952method}, linear
programming~\cite{dunstan1973greedy,dietrich2003greedy}, minimum
spanning trees~\cite{kruskal1956shortest,prim1957shortest}, partially
ordered sets~\cite{faigle1979greedy,dietrich2003greedy}, 
matroids~\cite{edmonds1971matroids,dress1990valuated},
greedoids~\cite{korte2012greedoids},
and so on,
perhaps culminating in the association between the greedy algorithm
and submodular functions~\cite{edmondspolyhedra, nemhauser1978,
  conforti1984submodular,fujishige2005submodular}.

Submodular functions have recently shown utility for a number of
machine learning and data science applications such as information
gathering \cite{krause2006near}, document summarization
\cite{lin2011class}, image segmentation \cite{kohli2009p3}, and string
alignment~\cite{lin2011word}, since such functions are natural for
modeling concepts such as diversity, information, and dispersion.
Defined over an underlying ground set $V$, a set function
$f:2^V\rightarrow \mathbb{R}$ is said to be submodular when for all
subsets $X,Y\subseteq V$, $f(X)+f(Y)\geq f(X\cup Y)+f(X\cap Y)$.
Defining $f(\set{v}|X)=f(\set{v}\cup X)-f(X)$ as the gain of adding
the item $v$ in the context of $X\subset V$, an equivalent
characterization of submodularity is via \emph{diminishing returns}:
$f(\set{v}|X) \geq f(\set{v}|Y)$, for any $X\subseteq Y\subset V$ and
$v\in V\setminus Y$.
A set function $f$ is monotonically non-deceasing if
$f(\set{v}|S)\geq 0$ for all $v\in V\setminus S$ and it is normalized
if $f(\emptyset)=0$.  In addition to being useful utility models,
submodular functions also have amiable optimization properties ---
many submodular optimization problems (both
maximization~\cite{wolsey1982analysis} and
minimization~\cite{cunningham1985submodular}) admit polynomial time
approximation or exact algorithms. Most relevant presently, the greedy
algorithm has a good constant-factor approximation guarantee, e.g.,
the classic \emph{$1-1/e$} and \emph{$1/(p+1)$} guarantees for
submodular maximization under a cardinality constraint or $p$ matroid
constraints~\cite{nemhauser1978analysis,fisher1978analysis}.

Certain subset selection problems in data science are not purely
submodular, however. For example, when choosing a subset of training
data in a machine learning system
\cite{wei2015-submodular-data-active}, there might be not only
redundancies but also complementarities amongst certain subsets of
elements, where the full collective utility of these elements are seen only when
utilized together. Submodular functions can only diminish, rather than
enhance, the utility of a data item in the presence other data
items. Supermodular set functions can model such phenomena, and are
widely utilized in economics and social sciences, where the notion of
complementary \cite{topkis2011supermodularity} is naturally needed,
but are studied and utilized less frequently in machine learning.
A set function $g(X)$ is said to be supermodular if $-g(X)$ is
submodular.

In this paper, we advance the state of the art in understanding when
the greedy (and the semigradient) algorithm offers a guarantee, in
particular for approximating the constrained maximization of an
objective that may be decomposed into the sum of a submodular and a
supermodular function (applications are given in
Section~\ref{sec:applications}). That is, we consider the following
problem
\begin{align}
  \text{Problem 1}.\quad\max_{X\in
\mathcal{C}} h(X) := f(X) + g(X), \label{eqn:submodular_supermodular_maximization}
\end{align}
where $\mathcal{C}\subseteq 2^V$ is a family of feasible sets, $f$ and
$g$ are normalized ($f(\emptyset)=0$), monotonic non-decreasing
($f(\set{s}|S)\geq 0$ for any $s\in V$ and $S\subseteq V$) submodular
and supermodular functions respectively\footnote{Throughout, $f$ \&
  $g$ are assumed monotonic non-decreasing submodular/submodular
  functions respectively.}  and hence are non-negative.  We call this
problem \emph{suBmodular-suPermodular} (BP) maximization, and $f+g$ a
\emph{BP function}, and we say $h$ admits a BP decomposition if $\exists f,g$
such that $h=f+g$ where $f$ and $g$ are defined as above.
In the paper, the set $\mathcal{C}$ may correspond either to a 
cardinality constraint
(i.e., $\mathcal{C} = \set{A\subseteq V \mid |A|\leq k}$ for some $k\geq
0$), %
or alternatively, a more general case where $\mathcal{C}$ is defined as
the intersection of $p$ matroids.  Hence, we may have
$\mathcal{C}=\{X\subseteq V \mid X\in \mathcal{I}_1\cap \mathcal{I}_2 \cap
\dots \cap \mathcal{I}_p \}$, where $\mathcal{I}_i$ is the set of
independent sets for the $i$th matroid
$\mathcal{M}_i=(V,\mathcal{I}_i)$.  A \emph{matroid} generalizes the
concept of independence in vector spaces, and is a pair
$(V,\mathcal{I})$ where $V$ is the ground set and $\mathcal{I}$ is a
family of subsets of $V$ that are \emph{independent} with the
following three properties: (1) $\emptyset\in \mathcal{I} $; (2)
$Y\in \mathcal{I}$ implies $X\in \mathcal{I}$ for all
$X\subseteq Y\subseteq V$; and (3) if $X, Y\in \mathcal{I}$ and
$|X|>|Y|$, then there exists $v\in X\setminus Y$ such that
$Y\cup \set{v} \in \mathcal{I}$. Matroids are often used as
combinatorial constraints, where a feasible set of an optimization
problem must be independent in all $p$ matroids.

The performance of the greedy algorithm for some special cases of BP
maximization has been studied before. For example, when $g(X)$ is modular,
the problem reduces to submodular maximization where, if $f$ and $g$
are also monotone, the greedy algorithm is guaranteed to obtain an
$1-1/e$ approximate solution under a cardinality
constraint~\cite{nemhauser1978analysis} and $\nicefrac{1}{p+1}$ for
$p$ matroids~\cite{fisher1978analysis,conforti1984submodular}. The
greedy algorithm often does much better than this in
practice. Correspondingly, the bounds can be significantly improved if
we also make further assumptions on the submodular function.  One such
assumption is the (total) \emph{curvature}, defined as
$\curv_f=1-\min_{v\in V}\frac{f(v|V\setminus \{v\})}{f(v)}$ --- the
greedy algorithm has a $\frac{1}{\curv_f}(1-e^{-\curv_f})$ and a
$\frac{1}{\curv_f +p}$ guarantee~\cite{conforti1984submodular} for a
cardinality and for $p$ matroid constraints,
respectively. Curvature is also attractive since it is linear time computable
with only oracle function access.
\citet{liu2017improved} shows that $\curv_f$ can be
replaced by a similar quantity, i.e.,
$b = 1-\min_{v\in A\in
  \mathcal{I}}\frac{f(\set{v}|A\setminus\set{v})}{f(\set{v})}$ for a
single matroid $\mathcal M = (V,\mathcal I)$, a quantity defined only
on the independent sets of the matroid, thereby improving the
bounds further. In the present paper, however, we utilize the traditional
definition of curvature. The current best guarantee is $1-\curv_f/e$
for a cardinality constraint using modifications of the continuous
greedy algorithm~\cite{sviridenko2015optimal} and
$\frac{1}{\epsilon+p}$ for multiple matroid constraints based on a
local search algorithm~\cite{lee2010submodular}. In another relevant
result, \citet{sarpatwar2017interleaved} gives a bound of
$\nicefrac{(1-e^{-(p+1)})}{(p+1)}$ for submodular maximization with a
single knapsack and the intersection of $p$ matroid constraints.

When $g(X)$ is not modular, the problem is much harder and is NP-hard
to approximate to any factor (Lemma~\ref{lemma:bpNotslovable}).  In
our paper, we show that bounds are obtainable if we make analogous
further assumptions on the supermodular function $g$.  That is, we
introduce a natural curvature notion to monotone non-decreasing
nonnegative supermodular functions, defining the {\bf supermodular
  curvature} as
$\curvg = \curv_{g(V)-g(V\setminus X)} = 1-\min_{v\in
  V}\frac{g(v)}{g(v|V\setminus \{v\})}$.  We note that $\curvg$ is
distinct from the steepness~\cite{il2001approximation,sviridenko2015optimal} of a
nonincreasing supermodular function (see Section~\ref{sec:curvature}).
The function $g(V)-g(V\setminus X)$ is a normalized monotonic
non-decreasing submodular function, known as the submodular function
dual to the supermodular function $g$
\cite{fujishige2005submodular}. \emph{Supermodular curvature} is a
natural dual to \emph{submodular curvature} and, like submodular
curvature, is computationally feasible to compute, requiring only
linear time in the oracle model, unlike other measures of non-submodularity
(Section~\ref{sec:relatedProblem}). Hence, given a BP decomposition of
$h=f+g$, it is possible, as we show below, to derive practical and
useful quality assurances based on the curvature of each component of
the decomposition.

We examine two algorithms, $\textsc{GreedMax}$ (Alg.~\ref{alg:via_greedy}) and $\textsc{SemiGrad}$ (Alg.~\ref{alg:ds}) and show that,
despite the two algorithms being different, 
both of them have a worst case guarantee of
$\frac{1}{\curv_f}\left[1-e^{-(1-\curv^g)\curv_f}\right]$ for a cardinality
constraint (Theorem~\ref{theorem:bound_card2}) and
$\frac{1-\curv^g}{(1-\curv^g)\curv_f + p}$ for $p$ matroid constraints
(Theorem~\ref{theo:greedyPmatroid}). If $\curv^g = 0$ (i.e., $g$ is modular),
the bounds reduce to $\frac{1}{\curv_f}(1-e^{-\curv_f})$ and $\frac{1}{\curv_f
	+p}$, which recover the aforementioned bounds. If $\curv^g = 1$ (i.e., $g$ is
fully curved) the bounds are $0$ since, in general, the problem is NP-hard to
approximate (Lemma~\ref{lemma:bpNotslovable}). For pure monotone supermodular
function maximization, the bounds yield $1-\curvg$ and $(1-\curvg)/p$ respectively.
We also show that no polynomial algorithm can
do better than $1-\curv^g+\epsilon$ or $(1-\curv^g)O(\frac{\ln p}{p})$ for
cardinality or multiple matroid constraints respectively unless P=NP. Therefore,
no polynomial algorithm can beat \textsc{GreedMax} by a factor of
$\frac{1+\epsilon}{1-e^{-1}}$ or $O(\ln(p))$ for the two constraints unless
P=NP.

\begin{table*}
  \label{table:results}
  \centering
  \begin{tabular}{|c|c|c|}
    \hline
    & bound & hardness  \\ \hline 
    cardinality constraint  &  $\frac{1}{\curv_f}\left[1-e^{-(1-\curv^g)\curv_f}\right]$      &  $1-\curv^g+\epsilon$ \\ \hline
    $p$ matroid constraints &  $\frac{1-\curv^g}{(1-\curv^g)\curv_f + p}$  & $(1-\curv^g)O(\frac{\ln p}{p})$ \\ \hline
  \end{tabular}
  \caption{Lower bounds for
    $\textsc{GreedMax}$ (Alg.~\ref{alg:via_greedy})/\textsc{SemiGrad} (Alg.~\ref{alg:ds})
    and BP maximization hardness.\looseness-1}
  \vspace{-1\baselineskip}
\end{table*}

\subsection{Applications}
\label{sec:applications}
Problem~\ref{eqn:submodular_supermodular_maximization} naturally
applies to a number of machine learning and data science applications.

\paragraph{Summarization with Complementarity} 

Submodular functions are an expressive set of models for summarization
tasks where they capture how data elements are mutually redundant. In
some cases, however, certain subsets might be usefully chosen
together, i.e., when their elements have a complementary
relationship. For example, when choosing a subset of training data
samples for supervised machine learning
system~\cite{wei2015-submodular-data-active}, nearby points on
opposite sides of a decision boundary would be more useful to
characterize this boundary if chosen together. Also, for the problem
of document summarization~\cite{lin2011class,lin2012learning}, where a subset of
sentences is chosen to represent a document, there are some cases
where a single sentence makes sense only in the context of other
sentences, an instance of complementarity.  In such cases, it is
reasonable to allow these relationships to be expressed via a monotone
supermodular function. One such complementarity family takes $g$ to be
a weighted sum of monotone convex functions composed with non-negative
modular functions, as in $g(A) = \sum_{i} w_i \psi_i(m_i(A))$.  A
still more expressive family includes the ``deep supermodular
functions'' \cite{bilmes-dsf-arxiv-2017} which consist of multiple
nested layers of such transformations.  A natural formulation of the
summarization with complementary problem is to maximize an objective
that is the weighted sum of a monotone submodular utility function and
one of the above complementarity functions.  Hence, such a formulation
is an instance of
Problem~\ref{eqn:submodular_supermodular_maximization}.  In either
case, the supermodular curvature is easy to compute, and for many
instances is less than unity leading to a quality assurance based on
the results of this paper.

\paragraph{Generalized Bipartite Matching}

Submodularity has been used to generalize bipartite matching.  For
example, a generalized bipartite matching~\cite{lin2011word} procedure
starts with a non-negative weighted bipartite graph $(V,U,E)$, where
$V$ is a set of left vertices, $U$ is a set of right vertices,
$E \subseteq V \times U$ is a set of edges, and
$h: 2^E \to \mathbb R_+$ is a score function on the edges.  Note that
a matching constraint is an intersection of two partition matroid
constraints, so a matching can be generalized to the intersection of
multiple matroid constraints. Word alignment between two sentences of
different languages~\cite{melamed2000models} can be viewed as a
matching problem, where each word pair is associated with a score
reflecting the desirability of aligning that pair, and an alignment is
formed as the highest scored matching under some
constraints. \citet{lin2011word} use a submodular objective functions
that can represent complex interactions among alignment
decisions. Also in \cite{bai-sgm-acm-bcb-2016}, similar bipartite
matching generalizations are used for the task of peptide
identification in tandem mass spectrometry.  By utilizing a BP
function in Problem~\ref{eqn:submodular_supermodular_maximization},
our approach can extend this to allow also for complementarity to be
represented amongst sets of matched vertices.

\subsection{Approach, and Related Studies}
\label{sec:relatedProblem}

An arbitrary set function can always be expressed as a difference of
submodular (DS) functions~\cite{narasimhanbilmes,rkiyeruai2012}.
Although finding such a decomposition itself can be hard
\cite{rkiyeruai2012}, the decomposition allows for additional
optimization strategies based on discrete semi-gradients
(Equation~\eqref{eq:semigrad}) that do not offer guarantees, even in
the unconstrained case~\cite{rkiyeruai2012}.  Our problem is a special
case of constrained DS optimization since a negative submodular
function is supermodular. Our problem also asks for a BP decomposition
of $h$ which is not always possible even for monotone functions
(Lemma~\ref{lemma:BPnotalways}). Constrainedly optimizing an arbitrary
monotonic non-deceasing set function is impossible in polynomial time
and not even approximable to any positive factor
(Lemma~\ref{lemma:bpNotslovable}).  In general, there are two ways to
approach such a problem: one is to offer polynomial time heuristics
without any theoretical guarantee (and hence possibly performing
arbitrarily poorly in worst case); another is to analyze (using
possibly exponential time itself, e.g., see below starting with the
submodularity ratio) the set function in order to provide theoretical
guarantees.  In our framework, as we will see, the BP decomposition
not only allows for additional optimization strategies as does a DS
decomposition, but also, given additional information about the
curvature of the two components (computable easily in linear time),
allows us to show how the set function can be approximately maximized
in polynomial time with guarantees.  With a curvature analysis, not
only the greedy algorithm but also a semi-gradient optimization
strategy (Alg.~\ref{alg:ds}) attains a guarantee even in the
constrained setting.  We also argued, in
Section~\ref{sec:applications}, that BP functions, even considering
their loss of expressivity relative to DS functions, are still quite
natural in applications.

\paragraph{Submodularity ratio and curvature}
\citet{2017arXiv170302100B} introduced a form of bound based on both
the submodularity ratio and introduced a generalized curvature.  The
submodularity ratio \cite{das2011submodular} of a non-negative set
function $h$ is defined as the largest scalar $\gamma$ s.t.\
$\sum_{\omega\in \Omega\setminus S}h(\Omega|S)\geq \gamma
h(\omega|S),\forall \Omega,S\subseteq V$ and is equal to one if and
only if $h$ is submodular. It is often defined as
$\gamma_{U,k}(h) = \min_{L\subseteq U, S:|S|\leq k, S\cap
  L=\emptyset}\frac{\sum_{x\in S}h(x|L)}{h(S|L)}$ for $U \subseteq V$
and $1\leq k \leq |V|$, and then $\gamma = \gamma_{V,|V|}(h)$. The
generalized curvature \cite{2017arXiv170302100B} of a non-negative set
function $h$ is defined as the smallest scalar $\alpha$ s.t.\
$h(i|S\setminus\set{i}\cup \Omega)\geq (1-\alpha)h(i|S\setminus
\set{i}), \forall \Omega,S\subseteq V, i \in S\setminus \Omega$.
\cite{2017arXiv170302100B} offers a lower bound of
$\frac{1}{\alpha}(1-e^{-\alpha \gamma})$ for the greedy algorithm.
Computing this bound is not computationally feasible in general
because both the submodularity ratio and the generalized curvature are
information theoretically hard to compute under the oracle model, as
we show in Section~\ref{sec:hardn-gener-curv}. This is unlike
curvatures $\curvf,\curvg$ which are both computable in linear time
given only oracle access to both $f$ and $g$.  We make further
comparisons between the pair $\curvf,\curvg$ with the submodularity
ratio in Section~\ref{sec:comp-subm-ratio}.

\paragraph{Approximately submodular functions}
A function $h$ is said to be $\epsilon$-approximately submodular if
there exists a submodular function $f$ such that
$(1-\epsilon)f(S) \leq h(S) \leq (1+\epsilon)f(S)$ for all subsets
$S$.
\citet{horel2016maximization} show that
the greedy algorithm achieves a $(1-1/e-O(\delta))$ approximation
ratio when $\epsilon = \frac{\delta}{k}$. Furthermore, this bound is
tight: given a $\nicefrac{1}{k^{1-\beta}}$-approximately submodular
function, the greedy algorithm no longer provides a constant factor
approximation guarantee.

\paragraph{Elemental Curvature and Total Primal Curvature}
\citet{wang2016approximation} analyze the approximation
ratio of the greedy algorithm on maximizing non-submodular functions under
cardinality constraints. Their bound is
$1-\left(1-\left(\sum_{i=1}^{k-1}\alpha^i\right)^{-1}\right)^k$ based on
the \emph{elemental curvature} with
$\alpha=\max_{S\subseteq X, i,j \in X}
\frac{f(i|S\cup\set{j})}{f(i|S)}$, and $\alpha^i$ the $i^\text{th}$ power
of $\alpha$.  
\citet{smith2017breaking} generalize this definition to
\emph{total primal curvature},
$\Gamma(x|B,A) = \frac{f(x|A\cup B)}{f(x|A)}$ and define an estimator $\hat{\Gamma}(i,S)$ satisfying
$\forall |T|\leq k,S\subset T,i=|T\setminus S|, x\notin T\cup
S:\Gamma(x|T,S)\leq \hat{\Gamma}(i,S)+\epsilon_i$. They claim
a bound of
$\left[1+\left(\frac{f(S^+)}{f(S)}-1\right)\sum_{t=0}^{k-1}(\hat{\Gamma}(t,S)+\epsilon_t)\right]^{-1}f(S^*)\leq
f(S)$ where $S$ is the greedy solution, and $S^+$ is the greedy solution
for an identical problem for $k+1$ cardinality constraints. They also claim
that finding a deterministic strict estimator $\hat{\Gamma}$ is not
feasible and therefore, they provide an algorithm for finding a
probabilistic estimator based on Monte-Carlo simulation.

\paragraph{Supermodular Degree}
Feige and et al.~\cite{feige2013welfare} introduce a parameter, the supermodular degree, for solving the welfare maximization problem. Feldman and et al.~\cite{feldman2014constrained,DBLP:conf_approx_FeldmanI14} use this concept to analyze monotone set function maximization under a $p$-extendable system constraint with guarantees.
A supermodular degree of one element $u\in V$ by a set function $h$ is defined as the cardinality of the set $\mathcal{D}^+_h(u)=\set{v\in V|\exists_{S\subseteq V} h(u|S+v)> h(u|S)}$, containing all elements whose existence in a set might increase the marginal contribution of $u$. The supermodular degree of $h$ is $\mathcal{D}^+_h=\max_{u\in V}|D^+_h(u)|$. A set system $(V, \mathcal{I})$ is called $p$-extendable~\cite{feldman2014constrained,DBLP:conf_approx_FeldmanI14} if for every two subsets $T\subseteq S\in \mathcal{I}$ and element $u\notin T$ for which $T\cup{u}\in \mathcal{I}$, there exists a subset $Y\subseteq S\setminus T$ of cardinality at most $p$ for which $S\setminus Y+u\in \mathcal{I}$, which is a generalization of the intersection of $p$ matroids. They offer a greedy algorithm for maximizing a monotonic non-decreasing set function $h$ subject to a $p$-extendable system with an guarantee of $\frac{1}{p(\mathcal{D}^+_h+1)+1}$ and time complexity polynomial in $n$ and $2^{\mathcal{D}^+_h}$~\cite{feldman2014constrained,DBLP:conf_approx_FeldmanI14}, where $n=|V|$. But again, $\mathcal{D}^+_h$ can not be calculated in polynomial time in general unlike our curvatures. Moreover, if we consider a simple supermodular function $g(X)=|X|^{1+\alpha}$ where $\alpha$ is a small positive number. Then $\mathcal{D}^+_h=n-1$ since all elements have supermodular interactions. Therefore, the time complexity of their algorithm is polynomial in $2^{n-1}$ and their bound is $\frac{1}{pn+1}$,  while our algorithm requires at most $n^2$ quires with a performance guarantee of $\frac{1-\log(n) \curv^g}{p}$ where $\curv^g = 1-\frac{1}{n^{1+\alpha}-(n-1)^{1+\alpha}}$. When $\alpha$ is small, our bound is around $n$ times betters than theirs; e.g., $n=10$, $p=5$, $\alpha=0.05$, ours is around $\frac{1}{7.61}$ while theirs is $\frac{1}{51}$.

\paragraph{Proportional Submodularity}
\citet{DBLP:journals/corr/BorodinLY14} define the notion of
proportionally submodular functions defined as those set functions $h$
satisfying
$|X|h(Y)+|Y|h(X) \geq |X \cap Y|h(X \cup Y) + |X \cup Y|h(X \cap Y)$
for all $X,Y \subseteq V$. The class of proportionally submodular
functions includes both submodular functions and also some
supermodular functions, although there are instances of BP functions,
e.g., $h(X) = |X|^4$, that are not proportionally submodular
(\cite{DBLP:journals/corr/BorodinLY14} proposition 3.12).

\paragraph{Discussion}
The above results are both useful and complementary with our analyses
below for BP-decomposable functions, thus broadening our understanding
of settings where the greedy and semi-gradient algorithms offer a
guarantee. We say our analysis is complementary in a sense the
following example demonstrates. Should a given function $h$ have a BP
decomposition $h=f+g$, then it is easy, given oracle access to both
$f$ and $g$, to compute curvatures and establish bounds. On the other
hand, if we do not know $h$'s BP decomposition, or if $h$ does not
admit a BP decomposition (Lemma~\ref{lemma:BPnotalways}), then we
would need to resort, for example, to the submodularity ratio and
generalized curvature bounds of \citet{2017arXiv170302100B}.

\section{Approximation Algorithms for BP Maximization}

\begin{minipage}[tb]{.52\linewidth}
	\begin{algorithm}[H]
		\caption{\textsc{GreedMax} for BP maximization}
		\begin{algorithmic}[1]
			\STATE \textbf{Input:} $f$, $g$ and constraint set $\mathcal{C}$.
			\STATE \textbf{Output: } An approximation solution $\hat{X}$.
			\STATE Initialize: $X_0 \leftarrow \emptyset$, $i\leftarrow 0$ and $R\leftarrow V$
			\WHILE{$\exists v\in R$ s.t.\ $X_i\cup v\in \mathcal{C}$}
			\STATE $v \in \argmax_{v\in R,X_i\cup v\in \mathcal{C}} f(v|X_i)+g(v|X_i)$.
			\STATE $X_{i+1} \leftarrow X_i \cup v$.
			\STATE $R \leftarrow R \setminus v$.
			\STATE $i\leftarrow i+1$.
			\ENDWHILE
			\STATE Return $\hat{X} \leftarrow X_i$.
		\end{algorithmic}
		\label{alg:via_greedy}
	\end{algorithm} 
\end{minipage}
\begin{minipage}[tb]{.48\linewidth}
	\begin{algorithm}[H]
		\caption{\textsc{SemiGrad} for BP maximization}
		
		\begin{algorithmic}[1]
			\STATE \textbf{Input:} $f$, $g$, constraint set $\mathcal{C}$ and an initial set $X_0$
			\STATE \textbf{Output: } An approximation solution $\hat{X}$.
			\STATE \textbf{Initialize:} $i\leftarrow 0$.
			\REPEAT
			\STATE pick a semigradient $g_i$ at $X_i$ of $g$
			\STATE $X_{i+1} \in\argmax_{X\in \mathcal{C}} f(X)+g_i(X)\backslash\backslash$ $\frac{1}{\curv_f}(1-e^{-\curv_f})-$Approximately solved by Algorithm~\ref{alg:via_greedy}
			\STATE $i\leftarrow i+1$
			\UNTIL{we have converged ($X_i=X_{i-1}$)}
			
			\STATE Return $\hat{X}\leftarrow X_i$
		\end{algorithmic}
		\label{alg:ds}
		
	\end{algorithm}
\end{minipage}
\paragraph{\textsc{GreedMax} (Alg.~\ref{alg:via_greedy})} The simplest
and most well known algorithm for approximate constrained non-monotone
submodular maximization is the greedy algorithm
\cite{nemhauser1978analysis}. We show that this also works boundedly
well for BP maximization when the functions are not both fully curved ($\curvf \leq 1, \curvg<1$).  At
each step, a feasible element with highest gain with respect to the
current set is chosen and added to the set. Finally, if no more
elements are feasible, the algorithm returns the greedy set.

\paragraph{\textsc{SemiGrad} (Alg.~\ref{alg:ds})}

Akin to convex functions, supermodular functions have tight modular
lower bounds. These bounds are related to the subdifferential
$\partial_g(Y)$ of the supermodular set function $g$ at a set
$Y \subseteq V$, which is defined
\cite{fujishige2005submodular}\footnote{\cite{fujishige2005submodular}
	defines the subdifferential of a submodular set function. The
	subdifferential definition for a supermodular set function takes the
	same form, although instances of supermodular subdifferentials
	(e.g., Eq.~\eqref{eq:supsubdifone}-\eqref{eq:supsubdiftwo}) take a
	form different than instances of submodular subdifferentials.} as:
\begin{align}
\partial_g(Y) &= \{y \in \mathbb{R}^n:g(X) - y(X) \geq g(Y) - y(Y)\;\text{for all } X \subseteq V\} %
\label{eq:semigrad}
\end{align}

It is possible, moreover, to provide specific
semigradients~\cite{rkiyersubmodBregman2012, rkiyersemiframework2013}
that define the following two modular lower bounds:
\begin{align}
	m_{g,X, 1}(Y) \triangleq g(X) - \!\!\!\! \sum_{j \in X \backslash Y } g(j| X \backslash j) + \!\!\!\! \sum_{j \in Y \backslash X} g(j| \emptyset)\scalebox{1.3}, 
	\label{eq:supsubdifone}
	\\
	m_{g,X, 2}(Y) \triangleq g(X) - \!\!\! \sum_{j \in X \backslash Y } g(j| V \backslash j) + \!\!\!\! \sum_{j \in Y \backslash X} g(j| X).
	\label{eq:supsubdiftwo}
\end{align}
\normalsize Then $m_{g,X, 1}(Y),m_{g,X, 2}(Y) \leq g(Y), \forall Y \subseteq V$ and $m_{g,X, 1}(X) = m_{g,X, 2}(X) = g(X)$. Removing constants yields normalized non-negative (since $g$ is monotone) modular functions for $g_i$ in Alg.~\ref{alg:ds}.

Having formally defined the modular lower bound of $g$, we are ready
to discuss how to apply this machinery to BP
maximization. \textsc{SemiGrad} consists of two stages. In the first
stage, it is initialized by an arbitrary set (e.g., $\emptyset$, $V$, or the solution of
\textsc{GreedMax}). In the second stage, \textsc{SemiGrad} replaces $g$
by its modular lower bound, and solves the
resulting problem using \textsc{GreedyMax}. The algorithm repeatedly
updates the set and calculates an updated modular lower bound until
convergence.\looseness-1

Since \textsc{SemiGrad} does no worse than the arbitrary initial set,
we may start with the solution of \textsc{GreedMax} and show that
\textsc{SemiGrad} is always no worse than
\textsc{GreedMax}. Interestingly, we obtain the same bounds
for \textsc{SemiGrad} even if we start with the empty set
(Theorems~\ref{theo:SubCard} and~\ref{theo:SubPmatroid}) despite 
that they may behave quite differently empirically
and yield different solutions 
(Section~\ref{sec:experiments}).

\section{Analysis of Approximation Algorithms for BP Maximization}
\label{sec:analys-appr-algor}

We next analyze the performance of two algorithms \textsc{GreedMax}
(Alg.~\ref{alg:via_greedy}) and \text{SemiGrad}(Alg.~\ref{alg:ds})
under a cardinality constraint and under $p$ matroid constraints.
First, we claim that BP maximization is hard and can not be
approximately solved to any factor in polynomial time in general.
\begin{restatable}{lem}{lembpNotslovable}
  \label{lemma:bpNotslovable}\cite{usul33967}
	There exists an instance of a BP maximization problem that can not be approximately solved to any positive factor in polynomial time.
\end{restatable}%
\vspace{\spacebeforeproofref}
\begin{proof}
  For completeness, Appendix~\ref{sec:bpNotslovable} offers a detailed proof
  based on \cite{usul33967}.
\end{proof}

It is also important to realize that not all monotone functions are
BP-decomposable, as the following demonstrates.
\begin{restatable}{lem}{lembpnotalways} There exists a monotonic non-decreasing set function $h$ that is not BP decomposable.
\label{lemma:BPnotalways}
\end{restatable}%
\vspace{\spacebeforeproofref}
\begin{proof}
See Appendix~\ref{sec:BPnotalways}.
\end{proof}

\subsection{Supermodular Curvature}
\label{sec:curvature}
Although BP maximization is therefore not possible in general, we show
next that we can get worst-case lower bounds using curvature whenever
the functions in question indeed have limited curvature.\looseness-1

The (total) curvature of a submodular function $f$ is defined as
$\curv_f=1-\min_{v\in V}\frac{f(v|V\setminus
  \{v\})}{f(v)}$~\cite{conforti1984submodular}. Note that
$0\leq\curv_f\leq 1$ since $0\leq f(v|V\setminus \{v\})\leq f(v)$ and
if $\curv_f=0$ then $f$ is modular. We observed that for any
monotonically non-decreasing supermodular function $g(X)$,
the dual submodular function~\cite{fujishige2005submodular}
$g(V)-g(V\setminus X)$ is always monotonically non-decreasing
and submodular. Hence, the definition of submodular curvature can be
naturally extended to supermodular functions $g$:
\begin{definition}
  \emph The \emph{supermodular curvature} of a non-negative monotone
  nondecreasing supermodular function is defined as
  $\curv^g=\curv_{g(V)-g(V\setminus X)}=1-\min_{v\in
    V}\frac{g(v)}{g(v|V\setminus\{v\})}$.
\end{definition}
For clarity of notation, we use a superscript for supermodular
curvature and a subscript for submodular curvature, which also
indicates the duality between the two. In fact, for
supermodular curvature, we can recover the submodular curvature.
\begin{corollary}
	$\curv_f=\curv^{f(V)-f(V\setminus X)}$.
\end{corollary}
The dual form also implies similar properties, e.g., we have that
$0\leq\curv^g\leq 1$ and if $\curv^g=0$ then $g$ is modular. In both
cases, a form of curvature indicates the degree of submodularity or
supermodularity. If $\curv_f=1$ (or $\curv^g=1$), we say that $f$ (or
$g$) is fully curved. Intuitively, a submodular function is very
(or fully) curved if there is a context $B$ and element $v$ at which the gain is close to (or equal to)
zero ($f(v|B) \approx 0$), whereas a supermodular function is very (or fully)
curved if there is an element $v$ whose valuation is close to (or equal to) zero
($g(v) \approx 0$).
We can calculate both submodular and
supermodular curvature easily in linear time. Hence, given a BP
decomposition of $h=f+g$, we can easily calculate both
curvatures, and the corresponding bounds, with only oracle access to
$f$ and $g$.
\begin{proposition}
	Calculating $\curv_f$ or $\curv^g$ requires at most $2|V|+1$ oracle
	queries of $f$ or $g$.
\end{proposition}

The steepness~\cite{il2001approximation,sviridenko2015optimal} of a
monotone nonincreasing supermodular function $g'$ is defined as
$s=1-\min_{v\in V}\frac{g'(v|V\setminus\set{v})}{g'(v|\emptyset)}$.
Here, the numerator and denominator are both negative and $g$ need not
be normalized. Steepness has a similar mathematical form to the
submodular curvature of a nondecreasing submodular function $f$, i.e.,
$\kappa_f=1-\min_{v\in
  V}\frac{f(v|V\setminus\set{v})}{f(v|\emptyset)}$, but is distinct
from the supermodular curvature. Steepness may be used to offer a
bound for the minimization of such nonincreasing supermodular
functions~\cite{sviridenko2015optimal}, whereas we in the present work
are interested in maximizing nondecreasing BP (and hence also
supermodular) functions.

\subsection{Theoretical Guarantees for \textsc{GreedMax}}
\label{sec:theor-guar-textscgr}

Before analyzing specific constraints, we first analyze each step of
\textsc{GreedMax} base on submodular and supermodular curvature.

The following holds for any chain of sets, not just those
produced by the greedy algorithm.
\begin{restatable}{lem}{lemgeneralGreed}
	\label{lemma:generalGreed}
	For any chain of solutions
	$\emptyset = S_0 \subset S_1 \subset \ldots \subset S_k$, where $|S_i|=i$, the following
	holds for all $i=0\ldots k-1$,
	\begin{align}
		h(X^*)\leq \curv_f \sum_{j:s_j\in S_i\setminus X^*} a_j+\sum_{j:s_j\in S_i\cap X^*}a_j
		+h(X^*\setminus S_i|S_i)
	\end{align}
	where  $\{ s_i\} =S_i\setminus S_{i-1}$, $a_i=h(s_i|S_{i-1})$ and $X^*$ is the optimal set.
\end{restatable}
 \vspace{\spacebeforeproofref}
\begin{proof}
	See Appendix~\ref{sec:generalGreed}.
\end{proof}

\subsubsection{Cardinality constraints}
\label{sec:card-constr}

In this section, we provide a lower bound for \textsc{Greedy}
maximization of a BP function under a cardinality constraint, inspired by the proof in~\cite{conforti1984submodular} where they focus only on submodular functions.

\begin{restatable}{lem}{lemboundcard}
	\label{lemma:bound_card1}
	\textsc{GreedMax} is guaranteed to obtain a solution $\hat{X}$ such that 
	\begin{align}
		h(\hat{X})\geq \frac{1}{\curv_f}\left[1-\left(1-\frac{(1-\curv^g)\curv_f}{k}\right)^k\right]h(X^*)
		\label{eqn:bound_card1}
	\end{align}
	where $X^*\in \argmax_{|X|\leq k} h(X)$, $h(X)=f(X)+g(X)$, $\curv_f$ is the curvature of submodular $f$ and $\curv^g$ is the curvature of supermodular $g$. 
\end{restatable}
\vspace{\spacebeforeproofref}
\begin{proof}
	See Appendix~\ref{sec:bound_card1}.
\end{proof}
\begin{restatable}{thm}{theoremboundcard}
	\label{theorem:bound_card2}
	{\bf Theoretical guarantee in the cardinality constrained case.}
	\textsc{GreedMax} is guaranteed to obtain a solution $\hat{X}$ such that 
	\begin{align}
          h(\hat{X})\geq \frac{1}{\curv_f}\left[1-e^{-(1-\curv^g)\curv_f}\right]h(X^*)
          \label{eqn:bound_card2}
	\end{align}
	where $X^*\in \argmax_{|X|\leq k} h(X)$, $h(X)=f(X)+g(X)$, $\curv_f$ is the curvature of submodular $f$ and $\curv^g$ is the curvature of supermodular $g$. 
\end{restatable}

\begin{proof}
	This follows Lemma~\ref{lemma:bound_card1} and uses the inequality $(1-\frac{a}{k})^k\leq e^{-a}$ for all $a\geq 0$ and $k\geq 1$.
\end{proof}

Theorem~\ref{theorem:bound_card2} gives a lower bound of \textsc{GreedMax} in terms of the submodular curvature $\curv_f$ and the supermodular curvature $\curv^g$. We notice that this bound immediately generalizes known results and provides one new one.
\begin{enumerate}[leftmargin=*,labelindent=0em,partopsep=-8pt,topsep=-8pt,itemsep=-2pt]
	\item $\curv_f = 0$, $\curv^g = 0$, $h(\hat{X})=h(X^*)$. In this case,
	the BP problem reduces to modular maximization under a cardinality
	constraint, which is solved exactly by the greedy algorithm.
	
	\item $\curv_f>0$, $\curv^g = 0$, $h(\hat{X})\geq \frac{1}{\curv_f}\left[1-e^{-\curv_f}\right]h(X^*)$. In this case, BP problem reduces to submodular maximization under a cardinality constraint, and  with the same $\frac{1}{\curv_f}\left[1-e^{-\curv_f}\right]$ guarantee for the greedy algorithm~\cite{conforti1984submodular}.
	
        \item If we take $\curvf \to 0$, we get
          $1-\curvg$, which is a new curvature-based bound for monotone supermodular
          maximization subject to a cardinality constraint.
          \label{item:one_minus_curvg_bound_card}

        \item $\curv^g = 1$, $h(\hat{X})\geq 0$ which means, in the general fully
	curved case for $g$, this offers no theoretical guarantee for
	constrained BP or supermodular maximization, consistent with \cite{usul33967} and
        Lemma~\ref{lemma:bpNotslovable}.
\end{enumerate}

\subsubsection{Weaker bound in the cardinality constrained case}

The bound in Equation~\eqref{eqn:bound_card2} is one of the major
contributions of this paper.  Another bound can be achieved using a
surrogate objective $h'(X)=f(X)+\sum_{v\in X} g(v)$, similar
to an approach used in \cite{curvaturemin}. We have that
$h'(X) \leq h(X)$ thanks to the supermodularity of $g$, and we can
apply \textsc{GreedMax} directly to $h'$, the solution of which has a
guarantee w.r.t.\ the original objective $h$. The proof of this bound
is quite a bit simpler, so we first offer it here immediately. On the other hand,
we also show that the bound obtained by this method is worse than
Equation~\eqref{eqn:bound_card2} for all $0<\curvf,\curvg<1$,
sometimes appreciably.

\begin{restatable}{lem}{simplebound}
	\label{lemma:simplebound}
	{\bf Weak bound in cardinality constrained case.} 
	\textsc{GreedMax} maximizing $h'(X)=f(X)+\sum_{v\in X} g(v)$ is guaranteed to obtain a solution $\hat{X}$ such that 
	\begin{align}
	h(\hat{X})\geq \frac{1-\curv^g}{\curv_f}\left[1-e^{-\curv_f}\right]h(X^*)
	\end{align}
	where $X^*\in \argmax_{|X|\leq k} h(X)$, $h(X)=f(X)+g(X)$, $\curv_f$ is the curvature of submodular $f$ and $\curv^g$ is the curvature of supermodular $g$. 
\end{restatable}
\begin{proof}
  According to lemma~\ref{lemma:curvInequ}~\ref{lemma:curvInequ4},
  $(1-\curvg)h(X)\leq h'(X)$ for all $X\subseteq V$. Also we have
  $h'(X)\leq h(X)$. And $h'$ is a monotone submodular function with
  $\curv_{h'}=1-\min_{v\in
    V}\frac{h'(v|V\setminus\set{v})}{h'(v)}=1-\min_{v\in
    V}\frac{f(v|V\setminus\set{v})+g(v)}{f(v)+g(v)}\leq 1-\min_{v\in
    V}\frac{f(v|V\setminus\set{v})}{f(v)}=\curvf$ since
  $0\leq f(v|V\setminus\set{v})\leq f(v)$.

  Using the traditional curvature bound for
  submodular maximization~\cite{conforti1984submodular},
  the greedy algorithm to maximize $h'$ 
  provides a solution $\hat{X}$ s.t.\
  $h'(\hat{X})\geq
  \frac{1}{\kappa_{h'}}\left[1-e^{-\kappa_{h'}}\right] h'(X^*)$ where
  $X^*\in \argmax_{|X|\leq k} h(X)$.
  Thus, we have
	\begin{align}
	h(\hat{X})&\geq h'(\hat{X})\geq \frac{1}{\curv_{h'}}\left[1-e^{-\curv_{h'}}\right] h'(X^*)  \geq \frac{1}{\curv_{f}}\left[1-e^{-\curv_{f}}\right] h'(X^*) \\
	&\geq\frac{1-\curv^g}{\curv_{f}}\left[1-e^{-\curv_{f}}\right] h(X^*)
	\end{align}
\end{proof}

Next, we show that this bound is almost everywhere worse than
Equation~\eqref{eqn:bound_card2}.

\begin{restatable}{lem}{simpleboundisworse}
	\label{lemma:simpleboundisworse}
$\frac{1}{\curv_f}\left[1-e^{-(1-\curv^g)\curv_f}\right]\geq \frac{1-\curv^g}{\curv_f}\left[1-e^{-\curv_f}\right]$ for all $0\leq \curvf,\curvg\leq 1$ where equality holds if and only if $\curv_f=0$ or $\curvg=0$ or $\curvg=1$. For simplicity, dividing by $0$ is defined using limits, e.g., $\frac{1}{\curv_f}\left[1-e^{-(1-\curv^g)\curv_f}\right]=\lim\limits_{\curvf\rightarrow 0^+}\frac{1}{\curv_f}\left[1-e^{-(1-\curv^g)\curv_f}\right]=1-\curvg$ when $\curv_f=0$.
\end{restatable}
\begin{proof}
	Let $\phi(\curvf,\curvg)=\frac{1}{\curv_f}\left[1-e^{-(1-\curv^g)\curv_f}\right]$ and $\psi(\curvf,\curvg)=\frac{1-\curv^g}{\curv_f}\left[1-e^{-\curv_f}\right]$. Specifically, $\phi(0,\curvg)=\lim\limits_{\curvf\rightarrow 0^+}\phi(\curvf,\curvg)=1-\curvg$ and $\psi(0,\curvg)=\lim\limits_{\curvf\rightarrow 0^+}\psi(\curvf,\curvg)=1-\curvg$. So if $\curvf=0$, $\phi(\curvf,\curvg)=\psi(\curvf,\curvg)$.
	
	When $0< \curvf\leq 1$, we notice that $\phi(\curvf,\curvg)=\psi(\curvf,\curvg)$ when $\curvg=0$ or $\curvg=1$. When $0<\curvg<1$, we have $\phi(\curvf,\curvg)>\psi(\curvf,\curvg)$ since $\phi(\curvf,\curvg)$ is a strictly concave function in $\curvg$ and $\psi(\curvf,\curvg)$ is linear in $\curvg$.
\end{proof}

A simple computation shows the maximum ratio of these two bounds is
$1/(1-e^{-1}) \approx 1.5820$ when $\curvf=1$ and $\curvg \to 1$.  As
another example, with $\curvf=1$ and $\curvg = \ln(e-1)\approx 0.541$, the
ratio is $\approx 1.2688$.

\subsubsection{Multiple matroid constraints}

Matroids are useful combinatorial objects for expressing constraints
in discrete problems, and which are made more useful when
taking 
the intersection of the independent sets of $p>1$ matroids defined
on the same ground set \cite{nemhauser1978analysis}. In this section,
we show that the greedy algorithm on a BP function subject to $p$
matroid independent constraints has a guarantee if $g$ is
not fully curved.\looseness-1

\begin{restatable}{thm}{theogreedyPmatroid}
	\label{theo:greedyPmatroid}
	{\bf Theoretical guarantee in the $p$ matroids case.} \textsc{GreedMax}  is guaranteed to obtain a solution $\hat{X}$ such that 
	\begin{align}
		h(\hat{X})\geq \frac{1-\curv^g}{(1-\curv^g)\curv_f + p}h(X^*)
	\end{align}
	where $X^*\in \argmax_{X\in \mathcal{M}_1\cap\ldots\cap\mathcal{M}_p} h(X)$, $h(X)=f(X)+g(X)$, $\curv_f$ is the curvature of submodular $f$ and $\curv^g$ is the curvature of supermodular $g$. 
\end{restatable}
\vspace{\spacebeforeproofref}
\begin{proof}
	See Appendix~\ref{sec:greedyPmatroid}.
\end{proof}
Theorem~\ref{theo:greedyPmatroid} gives a theoretical lower bound of
\textsc{GreedMax} in terms of submodular curvature $\curv_f$ and
supermodular curvature $\curv^g$ for the $p$ matroid constraints
case. Like in the cardinality case, this bound
also generalizes known results and yields a new one.
\begin{enumerate}[leftmargin=*,labelindent=0em,partopsep=-8pt,topsep=-8pt,itemsep=-2pt]
\item $\curv_f = 0$, $\curv^g = 0$,
  $h(\hat{X})\geq\frac{1}{p}h(X^*)$. In this case, the BP problem
  reduces to modular maximization under $p$ matroid
  constraints~\cite{conforti1984submodular}.

\item $\curv_f>0$, $\curv^g = 0$,
  $h(\hat{X})\geq \frac{1}{p+\curv_f}h(X^*)$ . In this case, the BP
  problem reduces to submodular maximization under $p$ matroid
  constraints~\cite{conforti1984submodular}.

\item If we take $\curvf \to 0$, we get
  $(1-\curvg)/p$, which is a new curvature-based bound for monotone supermodular
  maximization subject to a $p$ matroid constraints.
          
\item $\curv^g = 1$, $h(\hat{X})\geq 0$ which means that, in general,
there is no theoretical guarantee for constrained BP or supermodular
maximization. 
\end{enumerate}

\begin{figure}
	\begin{center}
		\begin{tabular}{ccc}
			\includegraphics[width=0.4\textwidth]{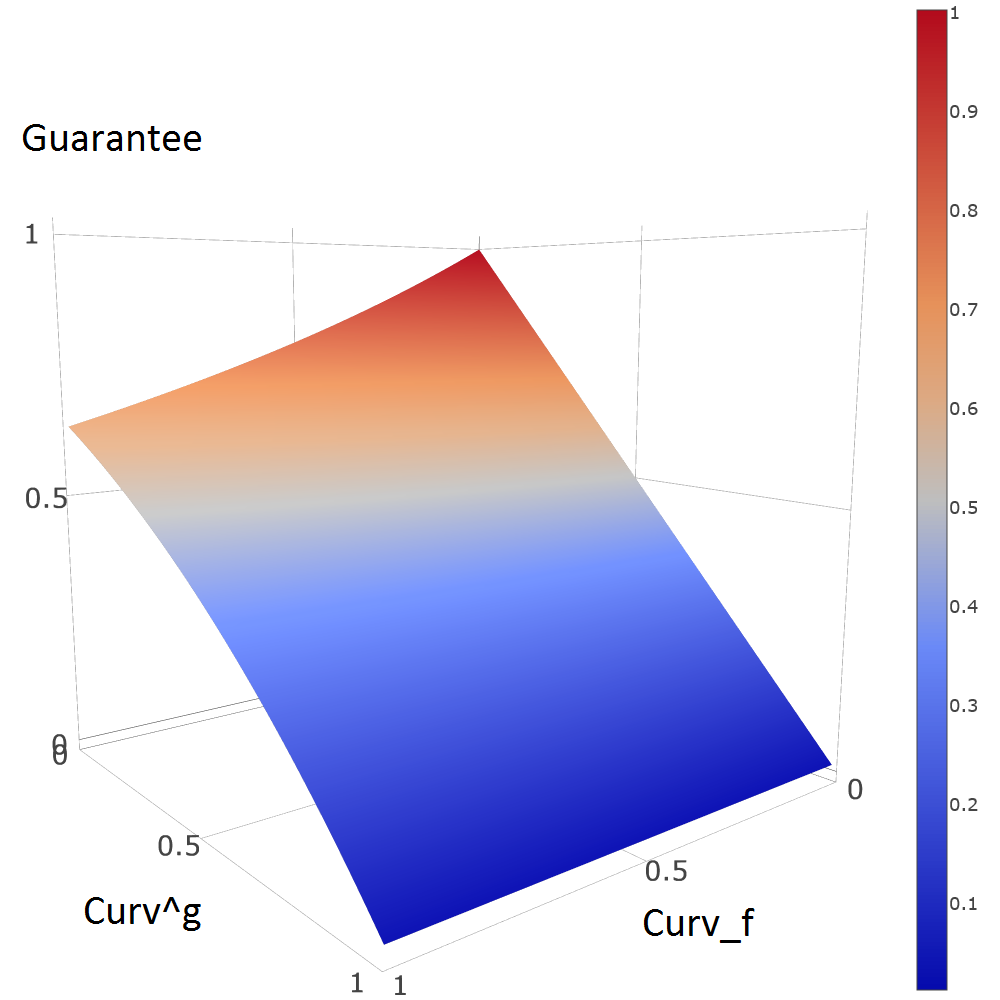} &
			\hspace{1em}                                                &
			\includegraphics[width=0.4\textwidth]{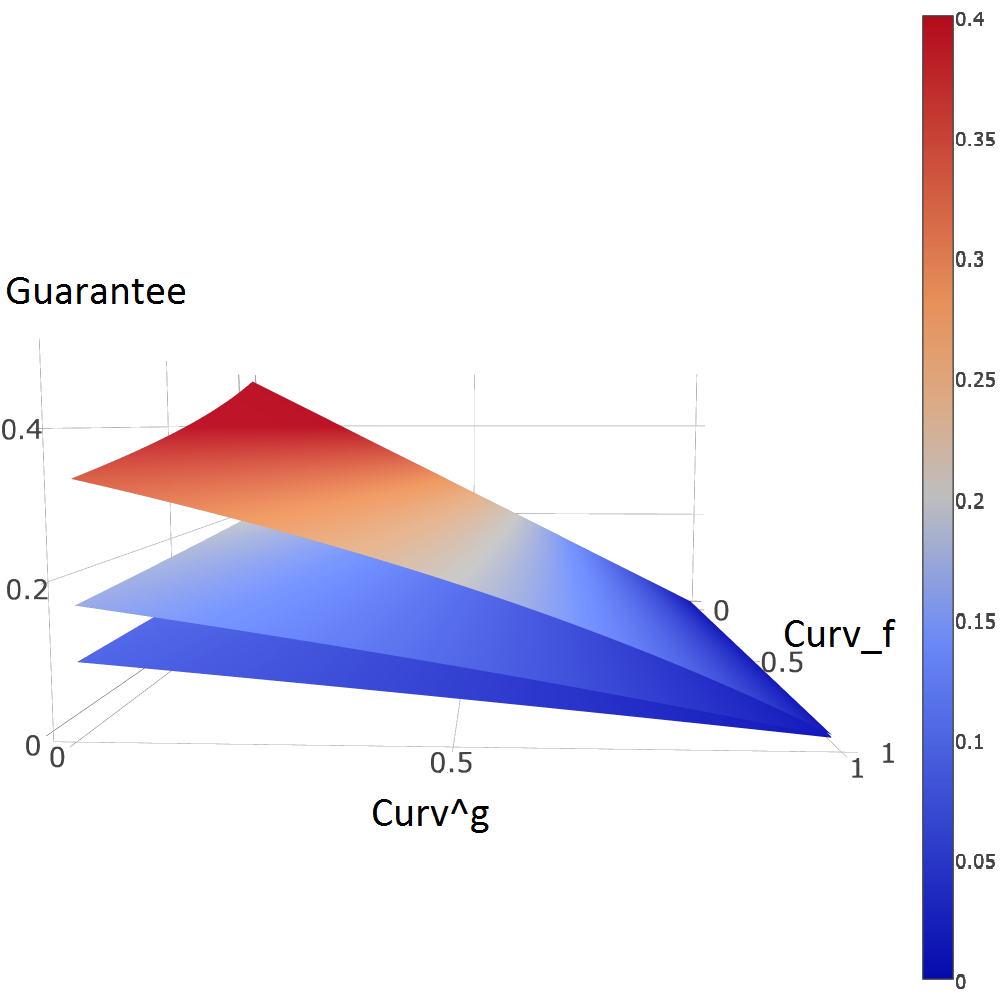}\\
			(a) Cardinality constraint & &(b) Multiple matroid constraints 
		\end{tabular}
	\end{center}
	\caption{Guarantees of \textsc{GreedMax} for two
		constraint types. The x and y axes are $\curv_f$ and
		$\curv^g$, respectively, and the z axis is the
		guarantee. In (b), from top to bottom, the surfaces
		represent $p = 2,5,10$.\looseness-1}
	\label{fig:theo_guarantee}
\end{figure}

\subsection{Theoretical guarantee of \textsc{SemiGrad}}

In this section, we show a perhaps interesting result that \textsc{SemiGrad}
achieves the same bounds as \textsc{GreedMax} even if we initialize
\textsc{SemiGrad} with $\emptyset$ and even though the two algorithms
can produce quite different solutions (as demonstrated in
Section~\ref{sec:experiments}).

\begin{restatable}{thm}{theoSubCard}
	\label{theo:SubCard}
	\textsc{SemiGrad} initialized with the empty set is guaranteed to obtain a solution $\hat{X}$ for the cardinality constrained case such that 
	\begin{align}
		h(\hat{X})\geq \frac{1}{\curv_f}\left[1-e^{-(1-\curv^g)\curv_f}\right]h(X^*)
	\end{align}
	where $X^*\in \argmax_{|X|\leq k} h(X)$, $h(X)=f(X)+g(X)$,
        \& $\curv_f$ (resp.\ $\curv^g$) is the curvature of $f$ (resp.\ $g$).\looseness-1
\end{restatable}
\vspace{\spacebeforeproofref}
\begin{proof}
	See Appendix~\ref{sec:SubCard}.
\end{proof}
\begin{restatable}{thm}{theoSubPmatroid}
	\label{theo:SubPmatroid}
	\textsc{SemiGrad} initialized with the empty set is guaranteed to obtain a solution $\hat{X}$, feasible for the $p$ matroid constraints,
	such that 
	\begin{align}
		h(\hat{X})\geq \frac{1-\curv^g}{(1-\curv^g)\curv_f + p}h(X^*)
	\end{align}
	where $X^*\in \argmax_{X\in \mathcal{M}_1\cap\ldots\cap\mathcal{M}_p} h(X)$, $h=f+g$,
        \& $\curv_f$ (resp.\ $\curv^g$) is the curvature of $f$ (resp.\ $g$).\looseness-1
\end{restatable}
\vspace{\spacebeforeproofref}
\begin{proof}
	See Appendix~\ref{sec:SubPmatroid}.
\end{proof}

All the above guarantees are plotted in
Figure~\ref{fig:theo_guarantee} (in the matroid case for $p=2$, $5$,
or $10$ matroids).

\section{Hardness}

We next show that the curvature $\curv^g$ limits the
polynomial time approximability of BP maximization.

\begin{restatable}{thm}{hardnessCard}
	\label{theo:hardnessCard}
	{\bf Hardness for cardinality constrained case.}  For all
	$0\leq\beta\leq 1$, there exists an instance of a BP function
	$h = f+ g$ with supermodular curvature $\curv^g=\beta$ such that no
	poly-time algorithm solving Problem 1 with a cardinality constraint
	can achieve an approximation factor better than
	$1-\curv^g+\epsilon$, for any $\epsilon >0$.
\end{restatable}
\vspace{\spacebeforeproofref}
\begin{proof}
	See Appendix~\ref{sec:hardnessCard}.
\end{proof}
For the $p$ matroid constraints case, \citet{hazan2006complexity} studied the complexity of
approximating $p$-set packing which is defined as follows: given a
family of sets over a certain domain, find the maximum number of
disjoint sets, which is actually a special case of finding the maximum
intersection of $p$ matroids. They claim that this problem cannot be
efficiently approximated to a factor better than
$O(\nicefrac{\ln p}{p})$ unless P = NP. We generalize their result to
BP maximization.

\begin{restatable}{thm}{hardnessPmatroid}
	\label{theo:hardnessPmatroid}
	{\bf Hardness for $p$ matroids constraint case.}	
	For all $0\leq\beta\leq 1$, there exists an instance of a
	BP function $h = f+ g$ with supermodular curvature $\curv^g=\beta$ such that no poly-time algorithm can achieve an approximation factor better than $(1-\curv^g)O(\frac{\ln p}{p})$  unless P=NP. 
\end{restatable}	
\vspace{\spacebeforeproofref}
\begin{proof}
	See Appendix~\ref{sec:hardnessPmatroid}.
\end{proof}

\begin{corollary}
	No polynomial algorithm can beat \textsc{GreedMax} or
	\textsc{SemiGrad} by a factor of
	$\frac{1+\epsilon}{1-e^{-1}}$ for cardinality, or $O(\ln(p))$
	for $p$ matroid constraints, unless P=NP.
\end{corollary}

\section{Computational Experiments}
\label{sec:experiments}

\begin{figure}	
\begin{center}
  \begin{tabular}{cc}
		\includegraphics[trim = 0 0.0cm 0 0cm,clip,width=0.4\textwidth]{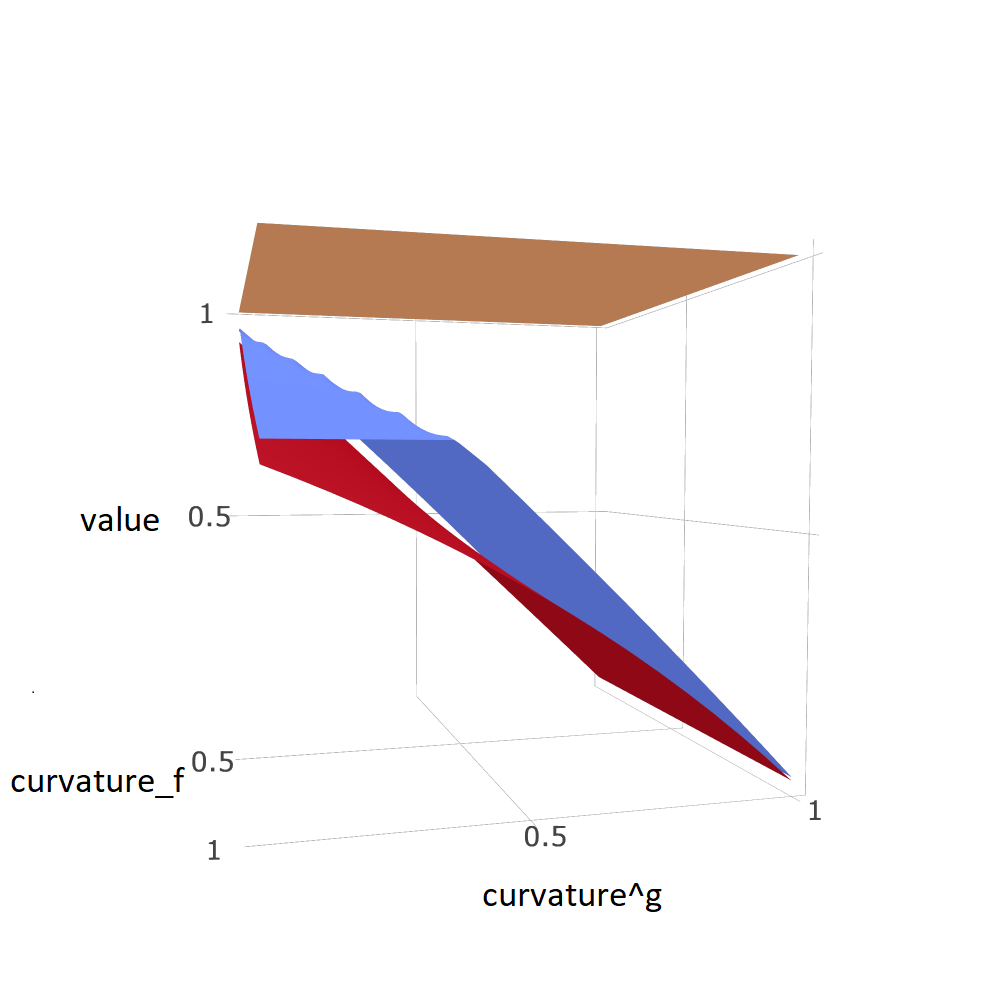} & 
		\includegraphics[trim = 0 0.0cm 0 0cm,clip,width=0.4\textwidth]{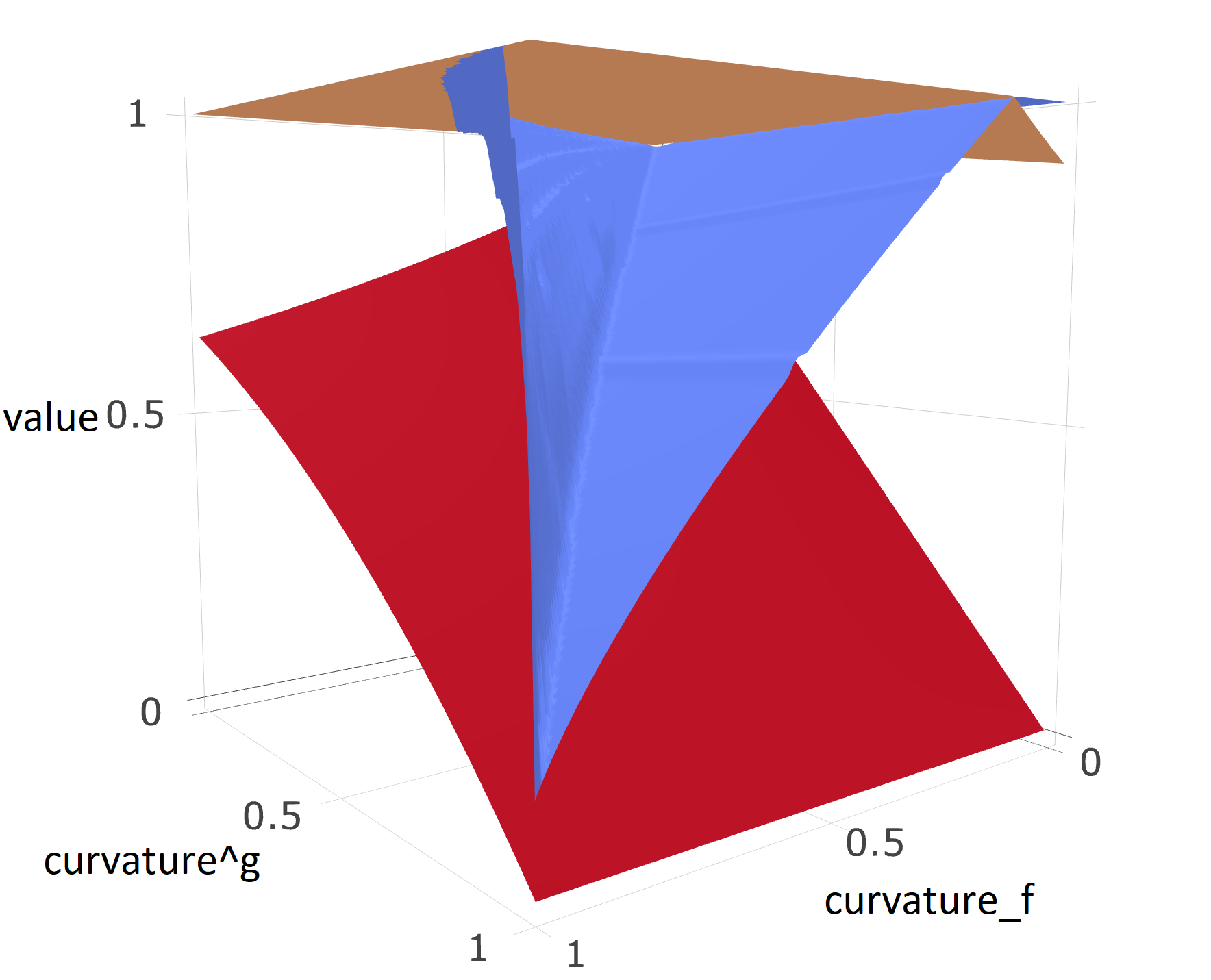} \\
		(a) & (b) 
  \end{tabular}
\end{center}
	\caption{Empirical test of our guarantee.
		The upper and middle surface
		indicate the performance of $\textsc{SemiGrad}$ and $\textsc{GreedMax}$ respectively, and the lower surface is the
		theoretical worst case guarantee. (a) and (b) are two sets of experiments.}
	\label{fig:experiment}
\end{figure} 

We empirically test our guarantees for BP maximization subject to a
cardinality constraint on contrived functions using
$\textsc{GreedMax}$ and $\text{SemiGrad}$. For the first experiment,
we let $|V| = 20$ set the cardinality constraint to $k = 10$, and
partition the ground set into $|V_1|=|V_2|=k$, $V_1\cup V_2 = V$ where
$V_1 = \set{v_1,v_2,\ldots,v_{k}}$. Let
$w_i =
\frac{1}{\alpha}\left[\left(1-\frac{\alpha}{k}\right)^i-\left(1-\frac{\alpha}{k}\right)^{i+1}\right]$
for $i = 1,2,\ldots,k$. Then we define the submodular and supermodular
functions as follows,
$f(X) = \left[\frac{k-\alpha|X\cap V_2|}{k}\right]\sum_{\{i:v_i\in
  X\}}w_i+\frac{|X\cap V_2|}{k}$,
$g(X) = |X| - \beta\min(1+|X\cap V_1|,|X|,k)+\epsilon\max(|X|,
|X|+\frac{\beta}{1-\beta}(|X\cap V_2|-k+1))$ and
$h(X) = \lambda f(X) + (1-\lambda) g(X)$ for
$0\leq \alpha,\beta,\lambda\leq 1$ and
$\epsilon=1\times10^{-5}$. Immediately, we notice that
$ \curv_f = \alpha$ and $\curv^g = \beta$. In particular, we choose
$\alpha,\beta,\lambda = 0,0.01,0.02,\ldots, 1$ and for all cases, we
normalize $h(X)$ using either exhaustive search so that $\text{OPT} = h(X^*)
=1$. Since we are doing a proof-of-concept experiment to verify the
guarantee, we are interested in the worst case performance at
curvatures $\curv_f$ and $\curv^g$. In Figure~\ref{fig:experiment}(a),
we see that both methods are always above the theoretical worst case
guarantee, as expected. Interestingly, $\textsc{SemiGrad}$ is doing
significantly better than $\textsc{GreedMax}$ demonstrating the
different behavior of the algorithms, despite their identical
guarantee.  Moreover, the gap between $\textsc{GreedMax}$ and the
bound layer is small (the maximum difference is 0.1852), which
suggests the guarantee for greedy may be almost tight in this
case.\looseness-1

The above example is designed to show the tightness of
$\textsc{GreedMax}$ and the better potential performance of
$\textsc{SemiGrad}$.  For a next experiment, we again let $|V| = 20$
and $k = 10$, partition the ground set into $|V_1|=|V_2|=k$,
$V_1\cup V_2 = V$. Let $f(X)=|X\cap V_1|^\alpha$ and
$g(X)=\max(0, \frac{|X\cap V_2|-\beta}{1-\beta})$
$0\leq \alpha,\beta \leq 1$, and normalize $h$ (by exhaustive search)
to ensure $\text{OPT} = h(X^*) = 1$.  Immediately, we notice that
the curvature of $f$ is $\curv_f= 1- k^\alpha + (k-1)^\alpha$ and the
curvature of $g$ is $\curv^g=\beta$. The objective BP function is
$h(X)=f(X)+g(X)$. We see that $\text{SemiGrad}$ is again doing better
than $\textsc{GreedMax}$ in most but not all cases
(Figure~\ref{fig:experiment}(b)) and both are above their bounds, as
they should be.

\newpage
\bibliography{submod}
\bibliographystyle{plainnat}
\looseness-1
\newpage
\appendix
\section{Proof of Lemma~\ref{lemma:bpNotslovable}}
\label{sec:bpNotslovable}
\lembpNotslovable*
\begin{proof}

  We consider the BP problem with ground set $n$ and a cardinality
  constraint $|X|\leq k = n/2$. Let $R\subseteq V$ be an arbitrary set
  with $|R| = k$. Let $f=0$ and $g'(X) = \max(|X|-k,0)$ so that
  $g'(X) = 0 $ for all $|X|=k$.  $g'(X)$ is clearly supermodular.
	
  Let $g(X) = g'(X)$ for all $X\neq R$ but $g(R) = 0.5$.  We notice
  that for $X \subset V$ and $v\notin X$, $g(v|X) = 0$ if
  $|X|\leq k-2$, $g(v|X) = 0$ or $0.5$ if $|X|= k-1$, $g(v|X) = 0.5$
  or $1$ if $|X|=k$, and $g(v|X) = 1$ if $|X|\geq k+1$. Immediately, we
  have for all $X\subset Y\subset V$ and $v\notin Y$,
  $g(v|X)\leq g(v|Y)$. Therefore, $g(X)$ is also supermodular.
	
  Next, we use a proof technique similar
  to~\cite{svitkina2008submodular}. Note that $g'(X)=g(X)$ if and only
  if $X\neq R$. So for any algorithm maximizing $g(X)$, before it
  evaluates $g(R)$, all function evaluations are the same with
  maximizing $g'(X)$.  Additionally, since $g'(X) = \max(|X|-k,0)$, it
  is permutation symmetric. Therefore, the algorithm can only do
  random search to find $R$. If the algorithm acquires a polynomial
  number $O(n^m)$ of sets of size $k$, the probability of finding
  $R$ is
  $\frac{O(n^m)}{\binom{n}{k}}\leq \frac{O(n^m)}{(n/k)^k} =
  \frac{O(n^m)}{2^{n/2}}\leq O(2^{-n/2+\epsilon n})$ for all
  $\epsilon>0$. Therefore, no polynomial time algorithm can
  distinguish $g$ and $g'$ with probability greater than
  $1-O(2^{-n/2+\epsilon n})$ and will return $0$ in almost all cases.
	
        Hence, we have $\max_{|X|\leq k} f(X)+g(X) = 0.5>0$ so no
        polynomial algorithm can do better than
        $\max_{|X|\leq k} f(X)+g'(X) = 0$ with high probability, or
        has any positive guarantee.
\end{proof}

\section{Proof of Lemma~\ref{lemma:BPnotalways}}
\label{sec:BPnotalways}
\lembpnotalways*
\begin{proof}
	Let $h(X)=\min(\max(|X|,1),3)-1$. This function is monotonic, and we wish to show it is not BP decomposable.  Let $A\subset B$ be subsets of $V$ with $|A|=1$ and $|B|=3$. Let $v\in V\setminus B$. We calculate that $h(v|\emptyset)=0$, $h(v|A)=1$, $h(v|B)=0$. So $h(v|\emptyset)+h(v|B)
	<h(v|A)$. 
	
	Assume $h(X)=f(X)+g(X)$ where $f$ is submodular, $g$ is supermodular and both are monotonic non-decreasing. We have $f(v|\emptyset)+f(v|B)\geq f(v|\emptyset)\geq f(v|A)$ and $g(v|\emptyset)+g(v|B)\geq g(v|B)\geq g(v|A)$. Therefore $h(v|\emptyset)+h(v|B)
	\geq h(v|A)$ by summing the two inequalities, which is a contradiction. We thus have that $h$ is not BP decomposable.
	  
\end{proof}

\section{Proof of Lemma~\ref{lemma:generalGreed}}
We begin with the following four-part lemma,
\begin{lemma}
	\label{lemma:curvInequ}
	For a BP function $h(X) = f(X) + g(X)$, we have 
	\begin{enumerate}[label=(\roman*)]
		\item $h(v|Y)\geq (1-\curv_f)h(v|X)$ for all
		$X\subseteq Y\subset V$ and $v\notin Y$ \label{lemma:curvInequ1}
		\item $h(v|Y)\leq \frac{1}{1-\curvg}h(v|X)$ for all
		$X\subseteq Y\subset V$ and $v\notin Y$ \label{lemma:curvInequ2}
		\item $h(X|Y)\geq (1-\curvf)\sum_{v\in X\setminus Y}h(v|Y)$  for
		all $X,Y\subseteq V$ \label{lemma:curvInequ3}
		\item $h(X|Y)\leq \frac{1}{1-\curvg}\sum_{v\in X\setminus Y}h(v|Y)$  for
		all $X,Y\subseteq V$ \label{lemma:curvInequ4}
	\end{enumerate}
\end{lemma}

\begin{proof}
	\begin{enumerate}[label=(\roman*)]
		\item $\curv_f = 1-\min_{v\in V}\frac{f(v|V\setminus\{v\})}{f(v)}$,
		therefore, $f(v|V\setminus\{v\})\geq (1-\curv_f)f(v)$ for all $v$.
		
		So we have $f(v|Y)\geq f(v|V\setminus\set{v})\geq (1-\curv_f)f(v)\geq  (1-\curv_f)f(v|X)$ and $g(v|Y)\geq g(v|X) \geq (1-\curv_f)g(v|X)$ for all $X\subseteq Y\subset V$ and $v\notin Y$. Therefore,  $h(v|Y)\geq (1-\curv_f)h(v|X)$ for all $X\subset Y\subseteq V$ and $v\notin Y$.
		\item  $\curvg = 1-\min_{v\in V}\frac{g(v)}{g(v|V\setminus\{v\})}$,
		therefore, $g(v|V\setminus\{v\})\leq \frac{1}{1-\curvg}g(v)$ for all $v$.
		
		So we have $g(v|Y)\leq g(v|V\setminus\set{v}) \leq  \frac{1}{1-\curvg}g(v)\leq \frac{1}{1-\curvg}g(v|X)$ and $f(v|Y)\leq f(v|X) \leq  \frac{1}{1-\curvg}f(v|X)$ for all $X\subseteq Y\subset V$ and $v\notin Y$. Therefore,  $h(v|Y)\leq \frac{1}{1-\curvg}h(v|X)$ for all $X\subset Y\subseteq V$ and $v\notin Y$.
		\item  Let $X\setminus Y$ be $\set{v_1,\ldots,v_m}$, $h(X|Y)=\sum_{i=1,2,\ldots,m}h(v_i|Y\cup\set{v_1}\cup\set{v_2}\cup\ldots\cup\set{v_{i-1}})\geq (1-\curvf)\sum_{i=1,2,\ldots,m}h(v_i|Y)=(1-\curvf)\sum_{v\in X\setminus Y}h(v|Y)$, according to \ref{lemma:curvInequ1}.
		\item  Let $X\setminus Y$ be $\set{v_1,\ldots,v_m}$,
		$h(X|Y)=\sum_{i=1,2,\ldots,m}h(v_i|Y\cup\set{v_1}\cup\set{v_2}\cup\ldots\cup\set{v_{i-1}})\leq \frac{1}{1-\curvg}\sum_{i=1,2,\ldots,m}h(v_i|Y)=\frac{1}{1-\curvg}\sum_{v\in X\setminus Y}h(v|Y)$, according to \ref{lemma:curvInequ2}.
	\end{enumerate}
\end{proof}
\label{sec:generalGreed}

\lemgeneralGreed*
\begin{proof}

  For any $i=0,\ldots,k-1$, we focus on the term $h(X^*\cup S_i)$.
	
  According to basic set operations,
  \begin{align}
    h(X^*\cup S_i)
    &=h(S_i)+h(X^*|S_i)= \\
    &= \sum_{j:s_j\in S_i\setminus X^*}
      a_j+\sum_{j:s_j\in S_i\cap X^*}a_j+h(X^*\setminus S_i|S_i).
  \end{align}

  We can also express $h(X^*\cup S_i)$ the other way around,
  $h(X^*\cup S_i)= h(X^*)+h(S_i\setminus X^*|X^*)$. Since we already
  have an order of element in $S_i$, we can expand
  $h(S_i\setminus X^*|X^*)$. When adding $s_j$ to the context 
  $S_{j-1}\cup X^*$ we do not need add elements that are not in
  $S_i\setminus X^*$
  since $h(s_j|X^* \cup S_{j-1}) = 0$ if $s_j \in X^*$. 
  Thus, using Lemma~\ref{lemma:curvInequ}~\ref{lemma:curvInequ1},
  we get
  $h(X^*\cup S_i)= h(X^*)+\sum_{j: s_j\in S_i\setminus X^*}
  h(s_j|X^*\cup S_{j-1} ) \geq h(X^*)+(1-\curv_f)\sum_{j:
    s_j\in S_i\setminus X^*} h(s_j|S_{j-1})$.
  
  Therefore, we have inequalities on both sides of $h(X^*\cup S_i)$ and
  we can join them together to get:
  \begin{align}
    h(X^*)+(1-\curv_f)\sum_{j: s_j\in S_i\setminus X^*} a_j
    & \leq \curv_f \sum_{j:s_j\in S_i\setminus X^*} a_j+\sum_{j:s_j\in S_i\cap X^*}a_j+h(X^*\setminus S_i|S_i), \\
    \intertext{ or }
    h(X^*)&\leq \curv_f\sum_{j: s_j\in S_i\setminus X^*} a_j+ \sum_{j: s_j\in S_i\cap X^*} a_j +h(X^*\setminus S_i|S_i).
  \end{align}
  
\end{proof}

\section{Proof of Lemma~\ref{lemma:bound_card1}}
\label{sec:bound_card1}

\lemboundcard*
\begin{proof}
  According to Lemma~\ref{lemma:generalGreed}, for all $i = 0, \dots, k-1$, 
  \begin{align}
    h(X^*)\leq \curv_f \sum_{j:s_j\in S_i\setminus X^*} a_j+\sum_{j:s_j\in S_i\cap X^*}a_j + h(X^*\setminus S_i|S_i)
  \end{align}
  
  Since \textsc{GreedMax} is choosing the feasible element with the
  largest gain, we have $h(v|S_i)\leq h(s_{i+1}|S_i)$ for all feasible
  $v\in X^*$. In fact, all elements in $X^*\setminus S_j$ are feasible
  since we are considering a cardinality constraint and
  $|S_j|\leq k-1$. Also, 
  $|X^*\setminus S_j| = |X^*| -|X^*\cap S_j| = k -|X^*\cap S_j|$, and
  therefore from Lemma~\ref{lemma:generalGreed}
  and Lemma~\ref{lemma:curvInequ}\ref{lemma:curvInequ4},
  we have that:
  \begin{align}
    h(X^*)
    \leq \curv_f \sum_{j: s_j \in S_i \setminus X^*}a_j
    + \sum_{j: s_j \in S_i\cap X^*}a_j
    +\frac{k-|X^*\cap S_i|}{1-\curv^g}a_{i+1}
    \label{eqn:constr_with_intersection_w_opt}
  \end{align}
  
  Next, we use a nested lemma, Lemma~\ref{lemma:carinality_inequ}, to get
  Equation~\eqref{eqn:bound_card1}.

  \begin{lemma}
    \label{lemma:carinality_inequ}
    Given any chain of solutions $\emptyset = S_0 \subset S_1 \subset \ldots \subset S_k$ such that $|S_i|=i$, if the following holds for all
    $i=0\dots k-1$:
    \begin{align}
      h(X^*)\leq \alpha \sum_{j:s_j\in S_i\setminus X^*} a_j+\sum_{j:s_j\in S_i\cap X^*}a_j
      +\frac{k-|X^*\cap S_i|}{1-\beta}a_{i+1}
      \label{eqn:carinality_inequ_ant}
    \end{align}
    where $0\leq \alpha,\beta\leq 1$ and $s_i=S_i\setminus S_{i-1}$,
    and $a_i=h(s_i|S_{i-1})$, then we have
    \begin{align}
      h(S_k)\geq\frac{1}{\alpha}\left[1-\left(1-\frac{(1-\beta)\alpha}{k}\right)^k\right]h(X^*).
      \label{eqn:carinality_inequ_conseq}
    \end{align}
    
  \end{lemma}
  
  \begin{proof}
    Assume $\beta < 1$ as otherwise the bound is immediate.
    This lemma aims to show one inequality
    (Equation~\eqref{eqn:carinality_inequ_conseq}) based on $k$
    other inequalities
    (Equation~\eqref{eqn:carinality_inequ_ant}) with $k$
    variables $a_1,\ldots,a_k$. In the inequalities,
    $s_j\in S_k\cap X^*$ and $s_j\in S_k\setminus X^*$ are not
    treated identically. We will, in fact, correspondingly treat
    the indices of the elements in $S_k\cap X^*$ as parameters. Recall,
    $S_k = \{ s_1, s_2, \dots, s_k\}$ is an ordered set and
    $S_k$ has index set $\{1, 2, \dots, k\} = [k]$.  Let
    $B=\{b_1,\ldots,b_p\} \subseteq [k]$ be the set of indices
    of $S_k\cap X^*$ where $b_i$'s are in increasing order (so
    $b_i < b_{i+1}$) and $p=|S_k \cap X^*|$. Thus, $i \in B$ means
    $s_i \in S_k\cap X^*$, and $i \in [k] \setminus B$ means
    $s_i \in S_k\setminus X^*$.

    Our next step is to view this problem as a set of
    parameterized (by $B$) linear programming problems. Each
    linear programming problem is characterized as finding:
    \begin{align}
      T(B) = T(b_1,b_2,\dots,b_p)=\min_{a_1,a_2,\dots,a_k}\sum_{i=1}^k a_i
    \end{align}
    subject to
    \begin{align}
      h(X^*)\leq \alpha \sum_{j \in [i-1] \setminus B_{i-1}} a_j+\sum_{j\in B_{i-1}}a_j
      +\frac{k-|B_{i-1}|}{1-\beta}a_{i}, \text{ for } i =1,\dots,k.
      \label{eqn:rhs_cnstr}
    \end{align}
    where $B_i=\{b\in B|b\leq i\}$.
    In this LP problem, $a_1,\ldots, a_k$ are non-negative
    variables, and $k,\alpha, \beta $ and $h(X^*)$ are
    fixed values.  Different indices $B = \set{ b_1,b_2,\dots,b_p }$
    define different LP problems, and our immediate goal
    is to show that
    $T(\emptyset)\leq T(b_1,b_2,\dots,b_p)$ for all
    $b_1,b_2,\dots,b_p$ and $p\geq 0$. In the below, we
    will use $\cnstr(B,a,i)$ to refer to the right hand
    side of Equation~\eqref{eqn:rhs_cnstr} for a given set
    $B$, vector $a$, and index $i=1, \dots, k$, and hence
    Equation~\eqref{eqn:rhs_cnstr} becomes
    $h(X^*) \leq \cnstr(B,a,i)$ for $i=1, \dots, k$.
    Note that $\cnstr(B,a,i)$ is linear in $a$ with
    non-negative coefficients.
    
    {\bf First}, we show that there exists an optimal
    solution\footnote{Optimal in this case means for the LP,
      distinct from the optimal BP
      maximization solution $X^*$.} $a_1,a_2,\ldots,a_k$ s.t.\ for
    all $r\leq k-1 $ with $r\in B$, $a_r\leq a_{r+1}$. Let
    $r_a$ be the largest $r$ s.t.\ $r\leq k-1$, $r\in B$
    and $a_r>a_{r+1}$; if such an $r$ does not exist, let
    $r_a=0$. Our goal here is equivalent to showing, for any
    feasible solution $\set{a_i}_{i=1}^k$ with $r_a>0$, we can
    create another feasible solution $\set{a_i'}_{i=1}^k$ with
    $r_{a'}=0$ and the objective
    $\sum_{i=1}^k a'_i\leq \sum_{i=1}^k a_i$. We do
    this iteratively, by in each step showing that
    for any feasible solution $\set{a_i}_{i=1}^k$ with $r_a>0$, we can create another feasible solution $\set{a_i'}_{i=1}^k$ with $r_{a'}\leq r_a-1$ and with objective having $\sum_{i=1}^k a'_i\leq \sum_{i=1}^k a_i$. Repeating this argument leads ultimately
    to $r_{a'} = 0$.

    Let $r=r_a$ for notational simplicity. Consider the $r^{\text{th}}$ and
    $(r+1)^{\text{th}}$ inequalities:
    \begin{align}
      \label{eq:line:1}
      h(X^*)&\leq \alpha \sum_{j \leq [r-1] \setminus B_{r-1}} a_j+\sum_{j\in B_{r-1}}a_j
              +\frac{k-|B_{r-1}|}{1-\beta}a_{r} \\
      \intertext{ and }
      h(X^*)&\leq \alpha \sum_{j \leq [r-1] \setminus B_{r-1}} a_j+\sum_{j\in B_{r-1}}a_j
              +a_r 
              +\frac{k-|B_{r-1}|-1}{1-\beta}a_{r+1}.
              \label{eq:line:2}
    \end{align}
    Since $a_r>a_{r+1}$ and $\beta < 1$,
    $\frac{k-|B_{r-1}|}{1-\beta}a_r >
    \frac{k-|B_{r-1}|-1}{1-\beta} a_{r+1} + a_r$ and thus
    the r.h.s.\ of Eq.~\eqref{eq:line:1} is always
    strictly larger than the r.h.s.\ of
    Eq.~\eqref{eq:line:2}. 

    Therefore,
    Eq.~\eqref{eq:line:1} is not tight and it is possible to decrease $a_r$ a little bit. Let $\set{a'_i}$ be another set of solutions with $a'_i=a_i$ for all $i=1,2,\ldots,r-1$; $a'_r=a_r-\epsilon$; $a'_i=a_i+\epsilon/(k-|B_r|)$ for $i=r+1,r+2,\ldots,k$ and $\epsilon = \left[1-\frac{1-\beta}{k-|B_{r-1}|}\right]\left[a_r-a_{r+1}\right]$. It is easy to see that $\epsilon> 0$ since $|B_{r-1}|\leq r-1\leq k-2$.
    
    Below, we show that $a'_r \leq a'_{r+1}$.  First, we
    notice $\sum_{i=1}^k a'_i\leq \sum_{i=1}^k a_i$ since
    $|B_r|\leq r$ and
    $-\epsilon + \frac{k-r}{k-|B_r|}\epsilon \leq 0$.
    Next, we want to show that $a'_1, a'_2, \ldots, a'_k$
    is still feasible. As mentioned above, define
    $\cnstr(B,a,i) =\alpha \sum_{j \in [i-1] \setminus
      B_{i-1}} a_j+\sum_{j\in B_{i-1}}a_j
    +\frac{k-|B_{i-1}|}{1-\beta}a_{i}$.
    
    We examine if $h(X^*)\leq \cnstr(B,a',i)$ or not for each $i$.
    
    \begin{enumerate}
    \item For $i = 1,2,\ldots, r-1$, $\cnstr(B,a',i)=\cnstr(B,a,i)\geq h(X^*)$.
    \item For $i = r$,
      $\cnstr(B,a',r)-\cnstr(B,a,r+1)
      = \frac{k-|B_{r-1}|}{1-\beta}\left[a_r-\epsilon\right]-a_r-\frac{k-|B_{r-1}|-1}{1-\beta}a_{r+1}
      \geq \frac{k-|B_{r-1}|}{1-\beta}\left[a_r-a_{r+1}\right]+a_{r+1}-a_{r}-\frac{k-|B_{r-1}|}{1-\beta}\epsilon
      = \left[\frac{k-|B_{r-1}|}{1-\beta}-1\right]
           \left[a_r-a_{r+1}\right]
           -\frac{k-|B_{r-1}|}{1-\beta}\left[1-\frac{1-\beta}{k-|B_{r-1}|}\right]
           \left[a_r-a_{r+1}\right]
           =0$. So $\cnstr(B,a',r)\geq \cnstr(B,a,r+1)\geq h(X^*)$.
    \item For $i = r+1,r+2,\ldots, k$, we compare $\cnstr(B,a',i)$ with $\cnstr(B,a,i)$. Note that $\cnstr(B,a,i) =\alpha \sum_{j \in [i-1] \setminus B_{i-1}} a_j+\sum_{j\in B_{i-1}}a_j
      +\frac{k-|B_{i-1}|}{1-\beta}a_{i}$ and it has three terms, that we consider individually.
      \begin{enumerate}

      \item The first term is not decreasing since $a'_i<a_i$ only if
        $i=r$, but $r\notin [i-1] \setminus B_{i-1}$. The increment therefore is
        at least 0.

      \item $a_r$ appears in the second term once, and when changing to $a'_r$, will
        decreases the value by $\epsilon$. However,
        $a'_j=a_j+\epsilon/(k-|B_r|)$ for all
        $j = r+1, r+2, \ldots, k$. Immediately, we notice the number
        of such $a_j$ in the second term is
        $\sum_{j\in B_{i-1}, j\geq r+1}1 = \sum_{j\in B_{i-1}, j\notin
          B_{r}}1 = |B_{i-1}| - |B_r|$. So the increment of the second
        term is
        $\frac{|B_{i-1}| - |B_r|}{k-|B_r|}\epsilon -\epsilon $.
      \item The third term is increased by
        $\frac{k-|B_{i-1}|}{(1-\beta)(k-|B_r|)}\epsilon\geq
        \frac{k-|B_{i-1}|}{k-|B_r|}\epsilon$.
      \end{enumerate}
      So overall, the increment is greater than or equal to
      $\frac{|B_{i-1}| - |B_r|}{k-|B_r|}\epsilon -\epsilon +
      \frac{k-|B_{i-1}|}{k-|B_r|}\epsilon \geq 0$, which means $\cnstr(B,a',i)\geq \cnstr(B,a,i)= h(X^*)$.
    \end{enumerate}
    Therefore, $\set{a'_i}_{i=1}^k$ still satisfies all the
    constraints but $\sum_{i=1}^{k}a'_k\leq \sum_{i=1}^{k}a_k$. Note that $r=r_a=\max(\set{r'\in B|r'\leq k-1,a_{r'}>a_{r'+1}})$ by definition. And we have $a'_i =a_i+\frac{\epsilon}{k-|B_r|}$ for $i = r+1,r+2,\ldots,k$. Therefore, $a'_{r'}\leq a'_{r'+1}$ for any $r'\in B\cap [r+1, k-1]$. Next we calculate $a'_{r}-a'_{r+1} = a_{r}-a_{r+1}-\epsilon-\frac{\epsilon}{k-|B_r|} = \left[1 -\left(1+\frac{1}{k-|B_r|}\right)\left(1-\frac{1-\beta}{k-|B_{r-1}|}\right)\right]\left(a_r-a_{r+1}\right)\leq 0$. Therefore, $a'_{r'}\leq a'_{r'+1}$ for all $r'\in B\cap [r, k-1]$ which implies $r_{a'}\leq r_{a}-1$. 
    
    By repeating the above steps, we can get a feasible solution
    $\set{a''}$ s.t.\ $r_{a''}=0$ and
    $\sum_{i=1}^{k}a''_k\leq \sum_{i=1}^{k}a_k$. Therefore, from any
    optimal solution $\set{a_i}_{i=1}^k$, we can also create another optimal
    solution $\set{a''_i }$ s.t.\ for all $r\in B$ and $r\leq k-1$, we
    have $a''_r\leq a''_{r+1}$. W.l.o.g, we henceforth consider only the optimal
    solutions $\set{a_i}_{i=1}^k$ with $r_a = 0$.
    
    {\bf Second}, we assume $r\in B$ but $r+1\notin B$ for some $r\leq
    k-1$. We can create $B'=B\cup\set{r+1}\setminus\set{r}$ and show
    for all $\set{a_i}_{i=1}^k$ that satisfies the constraints of $B$,
    $\set{a_i}_{i=1}^k$  will also still satisfy the constraints of $B'$ by
    showing that $\cnstr(B',a,i) \geq \cnstr(B,a,i)$ for $i = 1, \dots, k$.
    We consider each $i$ in turn.
    \begin{enumerate}
    \item For $i = 1,2,\ldots, r$, $\cnstr(B,a,i) = \cnstr(B',a,i)$.
    \item If $i = r+1$, we notice $a_r$ moves from the second term to the
      first, and the third term is changed from
      $\frac{k-|B_r|}{1-\beta}a_{r+1}$ to
      $\frac{k-|B'_r|}{1-\beta}a_{r+1}$ and $|B'_r| = |B_r|-1$. So the
      overall value is increased by
      $\cnstr(B',a,i) - \cnstr(B,a,i)
       = \frac{1}{1-\beta}a_{r+1} - (1-\alpha)a_{r}\geq 0$ since
      $a_r\leq a_{r+1}$.

    \item For $i = r+2,r+3,\ldots, k$, we notice that the third term does
      not change but $a_r$ moves from the second term to the first and
      $a_{r+1}$ moves from the first term to the second. Thus, the value
      is increased by $\cnstr(B',a,i) - \cnstr(B,a,i) = (1-\alpha)(a_{r+1}-a_r)\geq 0 $ since
      $a_r\leq a_{r+1}$.
      
    \end{enumerate}
    Since $\cnstr(B',a,i) \geq \cnstr(B,a,i)$ for $i = 1, \dots, k$,
    we have that $T(B')\leq T(B)$. Therefore, if we see two indexes in
    $B$ differ by at least 2, we can increase the first index by 1. 
    Repeating this process, we get 
    \begin{align}
      T(B)\geq T(k-p+1,k-p+2,\ldots,k).
    \end{align} 

    {\bf Third}, if $\set{a_i}_{i=1}^k$ satisfies the constraints
    for $B=\{k-p+1,k-p+2,\ldots,k\}$and $a_{k-p+1}\leq \ldots\leq a_k$, then 
    $\set{a_i}_{i=1}^k$ also must satisfy
    the constraints for $B'=\{k-p+2,k-p+3,\ldots,k\}$. We 
    show that  $\cnstr(B',a,i) \geq \cnstr(B,a,i)$ for $i = 1, \dots, k$
    and again consider each $i$ in turn.
    \begin{enumerate}
    \item For $i = 1,2,\ldots, k-p+1$, $\cnstr(B',a,i) = \cnstr(B,a,i)$.
    \item For $i = k-p+2,k-p+3,\ldots,k$, the change of the value is
      $\cnstr(B',a,i) - \cnstr(B,a,i) = (\alpha-1)a_{k-p+1}+\frac{1}{1-\beta}a_i$. We notice that
      $a_i\geq a_{k-p+1}$ since $k-p+1,k-p+2,\ldots,i-1\in B$. Thus,
      we have $\cnstr(B',a,i) - \cnstr(B,a,i) \geq 0$ and
      correspondingly $T(B)\geq T(B')$.
    \end{enumerate}
    
    Repeating this process, therefore, we have that
    \begin{align}
      T(B)\geq T(\emptyset)
    \end{align}

    Next, we calculate $T(\emptyset)$. For $B=\emptyset$ and any
    feasible
    (for Equation~\eqref{eqn:rhs_cnstr})
    $a_1,a_2,\ldots,a_k$, let $T_i$ be the partial sum 
    $T_i = \sum_{j= 1}^{i}a_j$ for $i=0, \dots, k$ with $T_0 = 0$. We get, for $i = 1, \dots, k$ that
    $h(X^*) \leq \cnstr(\emptyset,a,i)$ which takes the form
    \begin{align}
      h(X^*)&\leq \alpha \sum_{j\in [i-1]} a_j+\frac{k}{1-\beta} a_{i}, \\
      \intertext{which is the same as}
      h(X^*)&\leq \alpha T_{i-1}+\frac{k}{1-\beta} (T_i -T_{i-1}), \\
      \intertext{and also, after multiplying both sides by $(1-\beta)/k$ and then
      adding $(1/\alpha)h(X^*)$ to both sides,  the same as}
      \frac{1}{\alpha}h(X^*)-T_i&\leq \left(1-\frac{(1-\beta)\alpha}{k}\right)\left(\frac{1}{\alpha}h(X^*)-T_{i-1}\right). \\
    \end{align}
    We then repeatedly apply all $k$ inequalities from $i=k, \dots, 1$, to get
    \begin{align}
      \frac{1}{\alpha}h(X^*)-T_k&\leq \left(1-\frac{(1-\beta)\alpha}{k}\right)^k\left(\frac{1}{\alpha}h(X^*)-T_0)\right) \\
      \intertext{yielding}
      T_k&\geq\frac{1}{\alpha}\left[1-\left(1-\frac{(1-\beta)\alpha}{k}\right)^k\right]h(X^*).
    \end{align}
    
    Let
    $\gamma
    =\frac{1}{\alpha}\left[1-\left(1-\frac{(1-\beta)\alpha}{k}\right)^k\right]
    $.  So, for $B=\emptyset$ and any feasible $a_1,a_2,\ldots,a_k$,
    we have $\sum_{j=1}^{k}a_j = T_k \geq \gamma h(X^*)$. Therefore
    $T(\emptyset) = \min_{a_1,a_2,\ldots,a_k}\sum_{i=1}^k a_i \geq
    \gamma h(X^*)$.
    
    Recall that $T(B)\geq T(\emptyset)$ for all $B$. We thus have,
    with $a_i = h(s_i | \{ s_1, \dots, s_{i-1} \})$ (which are also
    feasible for Equation~\eqref{eqn:rhs_cnstr} with
    $B$ again the indices of $S_k \cap X^*$, which follows
    from Equation~\ref{eqn:constr_with_intersection_w_opt}),
    $h(S_k) = \sum_{i}^{k} a_i \geq T(B) \geq T(\emptyset) \geq \gamma
    h(X^*)$.
  \end{proof}
  Lemma~\ref{lemma:carinality_inequ} yields 
  Equation~\eqref{eqn:bound_card1} which shows the
  result for Lemma~\ref{lemma:bound_card1}.
\end{proof}

\section{Proof of Theorem~\ref{theo:greedyPmatroid}}
\label{sec:greedyPmatroid}
\theogreedyPmatroid*
\begin{proof}
  The greedy procedure produces a chain of solutions
  $S_0,S_1,\ldots,S_k$ such that $|S_i|=i$, $S_i\subset S_{i+1}$,
  where $k$ is the iteration after which any addition to $S_k$ is
  infeasible in at least one matroid, and hence\footnote{There should
    be no confusion here that the $k$ we refer to in this section is not
    any cardinality constraint, but rather the size of the greedy solution.}
  $|\hat X| = k$. Immediately, we notice all $S_i$ and $X^*$ are
  independent sets for all $p$ matroids.

  For $j = 0, \dots, k$ and $l = 1, \dots, p$, there exist at least
  $\max(|X^*|-j,0)$ elements $v \in X^* \setminus S_j$ s.t.\
  $v \notin S_j$ and $S_j + v \in \mathcal I(M_l)$, which follows from
  the third property in the matroid definition.  Therefore, for
  $j = 0, \dots, k-1$, $l = 1, \dots, p$, there are at most $j$
  elements of $X^*$ that can not be added to $S_j$.

  We next consider the intersection of all $p$ matroids.  For
  $j = 0, \dots, k$, since in each matroid, there are at most $j$
  elements of $X^*$ that cannot be added to $S_j$, the total possible
  number of elements for which there exists at least one matroid
  preventing us from adding to $S_j$ is $jp$ (the case that the $p$
  sets of at most $j$ elements are disjoint). In other words, there
  are at least $\max(|X^*|-pj,0)$ different $v\in |X^*|$ s.t.\
  $v\notin S_j$,
  $S_j\cup{\set{v}} \in \mathcal{M}_1\cap\ldots\cap\mathcal{M}_p$.

  We claim $|X^*|\leq pk$ as otherwise, by setting $j=k$ above, there
  are still feasible elements in $X^*\setminus S_k$ in the context of
  $S_k$, which indicates that $\textsc{GreedMax}$ has not ended at
  iteration $k$. Therefore, we are at liberty to create $pk-|X^*|$
  dummy elements, that are always feasible (i.e., independent in all
  matroids) and that have $h(v|X) = 0$ for all $X\subset V$ for each
  dummy $v$. We add these dummy elements to $X^*$ and henceforth
  assume, w.l.o.g., that $|X^*| = pk$.

  We next form an ordered $k$-partition of
  $X^*=X_0\cup X_1\cup \ldots\cup X_{k-1}$. We show below that it is possible to form this partition
  so that it has the following properties for $j=0,\dots, k-1$:
  \begin{enumerate}
  \item $|X_j|=p$;
  \item for all $v\in X_j$, we have $v\notin S_j$ and $S_j\cup{\set{v}} \in \mathcal{M}_1\cap\ldots\cap\mathcal{M}_p$ (i.e., $v$ can be added to $S_j$);
  \item and for all $j$ s.t.\ $s_{j+1}\in X^*\cap S_k$, we have $s_{j+1}\in X_j$.
  \end{enumerate}
  Immediately, we notice that property 3 is compatible with property 2.
  
  We construct this partition in an order reverse from that of the
  greedy procedure, that is we create $X_j$ from $j=k-1$ to $0$.
  Recall that, at each step with index $j=k-1,k-2,\ldots,0$, there are
  at least $|X^*| - pj = p(k-j)$ elements in $X^*$ can be added to $S_j$.

  When $j=k-1$, there are at least $p$ candidate
  elements\footnote{Elements that can be added at the given step.} in
  $X^*$ and we choose $p$ of them to form $X_{k-1}$. The element $s_k$
  can be added to $S_{k-1}$ because the greedy algorithm only adds
  feasible elements and hence, if also $s_k \in X^*$, then $s_k$ can
  be one of the elements in $X_{k-1}$. Thus, abiding property 3 above,
  we place $s_k\in X_{k-1}$.

  Continuing, for $j=k-2,k-3,\ldots,0$, there are at least $p$
  candidate elements in
  $X^*\setminus\left[X_{k-1}\cup X_{k-2}\cup\ldots\cup X_{j+1}\right]$
  since $|X_{k-1}\cup X_{k-2}\cup\ldots\cup
  X_{j+1}|=p(k-j-1)$ and we choose $p$ of them for $X_j$. 
  Moreover, if $s_{j+1}\in X^*$, we notice
  $s_{j+1}$ may be one of those candidate elements because of the greedy
  properties and since 
  $s_{j+1}\notin \left[X_{k-1}\cup X_{k-2}\cup\ldots\cup
    X_{j+1}\right]$ (this follows because $s_{j+1}\in S_{j'}$ for
    any $j'\geq j+1$, so $s_{j+1}$ is not a candidate element at step
    $j'= k-2,\ldots,j+1$). Similar to what was done in step $k-1$, we again choose $p$
    candidate elements to form $X_j$, and, if $s_{j+1}\in X^*$, we
    place $s_{j+1}\in X_{j}$.
	
    We then arrive at partition
    $X^*=X_0\cup X_1\cup \ldots\cup X_{k-1}$ with the aforementioned
    three properties.
	
    Next, we order the elements in $X^*=\set{x_1,\ldots,x_{pk}}$ where
    $\set{x_{jp+1},x_{jp+2},\ldots,x_{(j+1)p}}=X_j$ for
    $j=0,1,\ldots,k-1$.  According to greedy, we have
    $h(x_{jp+t}|S_j)\leq h(s_{j+1}|S_j) = a_{j+1}$ for
    $t = 1,\ldots, p$. Recall that $a_i$ is defined to be
    $ h(s_{i}|S_{i-1})$. Moreover, if $x_{jp+t} \in X^*\cap S_k$, we
    have $x_{jp+t} = s_{j+1}$.
	
	According to Lemma~\ref{lemma:generalGreed} above,
	\begin{align}
		h(X^*)&\leq \curv_f \sum_{j:s_j\in S_k\setminus X^*} a_j+\sum_{j:s_j\in S_k\cap X^*}a_j+h(X^*\setminus S_k|S_k)\\ \label{line:3}
		&= \curv_f \sum_{j:s_j\in S_k\setminus X^*} a_j+\sum_{j:s_j\in S_k\cap X^*}h(s_j|S_{j-1}) + \sum_{i = 1}^{pk}h(x_i|S_k\cup\set{x_1}\ldots\cup\set{x_{i-1}}){\bf 1}_{\set{x_i\in X^*\setminus S_k}}\\ \label{line:4}
		&\leq \curv_f \sum_{j:s_j\in S_k\setminus X^*} a_j+\frac{1}{1-\curv^g}\sum_{j:s_j\in S_k\cap X^*}h(s_j|S_{j-1}) + \frac{1}{1-\curv^g}\sum_{j = 0}^{k-1}\sum_{t=1}^p h(x_{jp+t}|S_j){\bf 1}_{\set{x_{jp+t}\in X^*\setminus S_k}}\\
		&= \curv_f \sum_{j:s_j\in S_k\setminus X^*} a_j\nonumber\\ \label{line:5} &\quad+\frac{1}{1-\curv^g}\left[\sum_{j:s_j\in S_k\cap X^*}h(s_j|S_{j-1}) + \sum_{j = 0}^{k-1}\sum_{t=1}^p h(x_{jp+t}|S_j)-\sum_{j = 0}^{k-1}\sum_{t=1}^p h(x_{jp+t}|S_j){\bf 1}_{\set{x_{jp+t}\in X^*\cap S_k}}\right]\\
		&= \curv_f \sum_{j:s_j\in S_k\setminus X^*} a_j\nonumber\\ \label{line:6}
		&\quad+\frac{1}{1-\curv^g}\left[\sum_{j:s_j\in S_k\cap X^*}h(s_j|S_{j-1}) + \sum_{j = 0}^{k-1}\sum_{t=1}^p h(x_{jp+t}|S_j)-\sum_{j:s_j\in S_k\cap X^*}h(s_j|S_{j-1})\right]\\\label{line:7}
		&\leq \curv_f \sum_{j:s_j\in S_k\setminus X^*} a_j+ \frac{1}{1-\curv^g}\sum_{j = 0}^{k-1}\sum_{t=1}^p a_{j+1}\\
                      &\leq \left[\curv_f + \frac{p}{1-\curv^g}\right]\sum_{j = 0}^{k-1} a_{j+1}
                        =  \left[\curv_f + \frac{p}{1-\curv^g}\right]h(\hat X)
	\end{align}
	where ${\bf 1}_{\set{\text{condition}}}$ equals 1 if the condition is met and is 0 otherwise.  Line \ref{line:3} to \ref{line:4} hold because of Lemma~\ref{lemma:curvInequ}~\ref{lemma:curvInequ2}. As for Line \ref{line:5} to \ref{line:6}, we notice  $x_{jp+t} = s_{j+1}$ if $x_{jp+t}\in X^*\cap S_k $. Line \ref{line:6} to line \ref{line:7} follows via the greedy procedure. 
	
	Therefore, we have our result which is
	\begin{align}
		h(\hat{X})&\geq \frac{1-\curv^g}{(1-\curv^g)\curv_f + p}h(X^*).
	\end{align}
\end{proof}
\section{Proof of Theorem~\ref{theo:SubCard}}
\label{sec:SubCard}
\theoSubCard*

\begin{proof}
	
	If \textsc{SemiGrad} is initialized by empty set, we need to calculate the semigradient of $g$ at $\emptyset$. By definition, we have
	
	\begin{align}
		m_{g,\emptyset, 1}(Y) = m_{g,\emptyset, 2}(Y) = \sum_{v \in Y} g(j)
	\end{align}
	
	So in the first step of \textsc{SemiGrad}, we are optimizing $h'(X)=f(X)+m_g(X) = f(X) + \sum_{v \in X} g(v)	$ by \textsc{GreedMax}. We will focus elusively on this step as later iterations can only improve the objective value.
	
	According to Lemma~\ref{lemma:generalGreed}, we have 
	\begin{align}
		h(X^*)\leq \curv_f \sum_{j:s_j\in S_i\setminus X^*} h(s_{j}|S_{j-1})+\sum_{j:s_j\in S_i\cap X^*}h(s_{j}|S_{j-1}) + h(X^*\setminus S_j|S_j)
	\end{align}
	
	Since \textsc{GreedMax} is choosing the feasible element with the largest gain, in the semigradient approximation we have $h'(v|S_i)\leq h'(s_{i+1}|S_i)$ instead of $h(v|S_i)\leq h(s_{i+1}|S_i)$. We get:
	\begin{align}
		h(X^*\setminus S_j|S_j) &= f(X^*\setminus S_j|S_j) + g(X^*\setminus S_j|S_j)\\
		&\leq \sum_{v\in X^*\setminus S_j}f(v|S_j) + \frac{1}{1-\curv^g}\sum_{v\in X^*\setminus S_j} g(v)\\
		& \leq \frac{1}{1-\curv^g}\sum_{v\in X^*\setminus S_j} h'(v|S_j) \\
		& \leq \frac{1}{1-\curv^g}\sum_{v\in X^*\setminus S_j} h'(s_{j+1}|S_j) \\
		& = \frac{1}{1-\curv^g}\sum_{v\in X^*\setminus S_j} f(s_{j+1}|S_j) + g(s_{j+1}) \\
		& \leq \frac{1}{1-\curv^g}\sum_{v\in X^*\setminus S_j} f(s_{j+1}|S_j) + g(s_{j+1}|S_j) \\
		& = \frac{|X^*\setminus S_j|}{1-\curv^g} h(s_{j+1}|S_j)
	\end{align}
And hence,
	\begin{align}
		h(X^*)\leq \curv_f \sum_{j : s_j \in S_i\setminus X^*}a_i+\sum_{j: s_j \in S_i\cap X^*}a_i
		+\frac{k-|X^*\cap S_i|}{1-\curv^g}s_{i+1}.
	\end{align}
We can then use Lemma~\ref{lemma:carinality_inequ} to $h$ to finish the proof.
\end{proof}
\section{Proof of Theorem~\ref{theo:SubPmatroid}}
\label{sec:SubPmatroid}
\theoSubPmatroid*

\begin{proof}
	If \textsc{SemiGrad} is initialized by empty set, we need to calculate the semigradient of $g$ at $\emptyset$. By definition, we have
	
	\begin{align}
		m_{g,\emptyset, 1}(Y) = m_{g,\emptyset, 2}(Y) = \sum_{v \in Y} g(j)
	\end{align}
	
	So in the first step of \textsc{SemiGrad}, we are optimizing
	$h'(X)=f(X)+m_g(X) = f(X) + \sum_{v \in X} g(v) $ by
	\textsc{GreedMax}. We will focus on this step.
	
	According to Lemma~\ref{lemma:generalGreed}, we have 
	\begin{align}
		h(X^*)\leq \curv_f \sum_{j:s_j\in S_i\setminus X^*} h(s_{j}|S_{j-1})+\sum_{j:s_j\in S_i\cap X^*}h(s_{j}|S_{j-1}) + h(X^*\setminus S_j|S_j)
	\end{align}
	
	We then follow the proofs of
        Theorems~\ref{theo:greedyPmatroid} and~\ref{theo:SubCard}.
        The only difference is that
        in Theorem~\ref{theo:greedyPmatroid} we have
        $h(v|S_i)\leq h(s_{i+1}|S_i)$ for all feasible $v$, but in
        this proof, we have $h'(v|S_i)\leq h'(s_{i+1}|S_i)$, which
        does not affect the proof as shown in the proof of
        Theorem~\ref{theo:SubCard}.
\end{proof}

\section{Proof of Theorem~\ref{theo:hardnessCard}}
\label{sec:hardnessCard}
\begin{lemma}
	\label{lemma:hardnessexample}
	(lemma 4.1 from \cite{svitkina2008submodular}) Let $R$ be a random subset of $V$ of size $\alpha = \frac{x\sqrt{n}}{5}$, let $\beta=\frac{x^2}{5}$, and let $x$ be any parameter satisfying $x^2=\omega(\ln n)$ and such that $\alpha$ and $\beta$ are integer. Let $f_1(X)=\min(|X|,\alpha)$ and $f_2(X)=\min(\beta+|X\cap \bar{R}|,|X|,\alpha)$. Any algorithm that makes a polynomial number of oracle queries has probability $n^{-\omega(1)}$ of distinguishing the functions $f_1$ and $f_2$.
\end{lemma}
\hardnessCard*
\begin{proof}
	$\curv^g=\alpha=0$ is trivial since no algorithm can do better than 1.
	
	The case when $\curv^g=1$ can be proven using the example in
	Lemma~\ref{lemma:bpNotslovable}. $g(X)=\max\{|X|-k,0\}$, except for
	a special set $R$ where $g(R)=0.5$ and $|R|=k$.
	
	For the other case, we prove this result using the hardness
        construction from~\cite{goemans2009approximating,
          svitkina2008submodular}. The intuition is to construct two
        supermodular functions, $g$ and $g'$ both with curvature $\curv^g$
        which are indistinguishable\footnote{Indistinguishable means
          for all sets $X$ that the algorithm evaluates, $g(X)=g'(X)$.}
        with high probability in polynomially many function
        queries. Therefore, any polynomial time algorithm to
        maximize $g(X)$ can not find $\hat{X}\subseteq V$ with
        $|\hat{X}|\leq k$ s.t.\ $g(\hat{X})>\max_{X\leq k}g'(X)$;
        otherwise we will have
        $g(\hat{X})>\max_{X\leq k}g'(X)\geq g'(\hat{X})$ which
        contradicts the indistinguishability. In this case, the
        approximate ratio
        $\frac{g(\hat{X})}{OPT}\leq \frac{\text{OPT}'}{\text{OPT}}$
        where $\text{OPT}=\max_{X\leq k}g(X)$ and
        $\text{OPT}'=\max_{X\leq k}g'(X)$. The guarantee, by
        definition, is the best case approximate ratio and, thus no
        greater than $\frac{\text{OPT}'}{\text{OPT}}$. If any
        polynomial algorithm has a guarantee greater than
        $\frac{\text{OPT}'}{\text{OPT}}$, then it contradicts the
        information theoretic hardness. This is meaningful
        if $\text{OPT}' < \text{OPT}$.
	
	Let $g(X)=|X|-\beta\min\{\gamma + |X \cap \bar{R}|, |X|,
        \alpha\}$ and $g'(X)=|X|-\beta\min\{|X|,\alpha\}$ , where
        $R \subseteq V$ is a random set of cardinality $\alpha$. Let
        $\alpha = x\sqrt{n}/5$ and $\gamma = x^2/5$ and let $x$ be any
        parameter satisfying $x^2=\omega(\ln n)$ s.t.\ $\gamma<\alpha$
        are positive integers and
        $\alpha\leq \frac{n}{2}-1$.\footnote{These examples and the
          specific parameters like 5 are adopted from
          \cite{svitkina2008submodular}.} $g$ and $g'$ are modular
        minus submodular functions, which implies
        supermodularity. Monotonicity follows from
        $g(v|X),g'(v|X)\geq 0$. Also, $\text{OPT} = \alpha - \beta \gamma
        > \text{OPT}' = \alpha(1-\beta)$.
	
	Next, we calculate the supermodular curvature. $g(\emptyset)=g'(\emptyset)=0$. $g(v)=g'(v)=1-\beta$ for all $v\in V$ since $\alpha,\gamma\geq 1$. $g(V\setminus \set{v})=g'(V\setminus \set{v})=n-1-\beta\alpha$ and $g(V)=g'(V)=n-\beta\alpha$ for all $v\in V$ since $\alpha\leq\frac{n}{2}-1$.	Therefore, $\curv^g=1-\min_{v\in V}\frac{g(v)}{g(v|V-{v})}=\beta$. $\curv^{g'}=1-\min_{v\in
		V}\frac{g'(v)}{g'(v|V-{v})}=\beta$. So $g$ and $g'$ are monotone
	non-decreasing supermodular functions with curvature $\beta$. Let $f(X)=0$ for all $X$ and $h(X)=f(X)+g(X)=g(X)$ is the objective BP function.
	
	Any algorithm that uses a polynomial number of queries can
        distinguish $g$ and $g'$ with probability only
        $n^{-\omega(1)}$ according to
        lemma~\ref{lemma:hardnessexample}~\cite{svitkina2008submodular}.
        More precisely, $g(X)> g'(X)$\footnote{Note that
          $g(X)\geq g'(X)$ for all $X\subseteq V$ for any $\alpha$ and
          $\gamma$.} if and only if $\gamma+|X\cap \bar{R}|< |X|$ and
        $\gamma+|X\cap \bar{R}| < \alpha$. It is equivalent with
        asking $|X\cap R|>\gamma$ and $|X\cap
        \bar{R}|<\alpha-\gamma$. Moreover,
        $\text{Pr}(g(X)\neq g'(X))$, where randomness is over random
        subsets $R \subseteq V$ of size $\alpha$, is maximized when
        $|X|=\alpha$~\cite{svitkina2008submodular}.  In this case, the
        two conditions become identical, and since
        $|X| = |X \cap \bar R| + |X \cap R|$, the condition 
        $g(X)> g'(X)$ happens when only $|X\cap R|>\gamma$. Intuitively,
        $E|X\cap R|= \frac{\alpha^2}{n}=\frac{\gamma}{5}$ where
        $R$ is a random set (of arbitrary size) and $X$ is an arbitrary
        but fixed set of size $\alpha$. So $|X\cap R|$ is
        located in small interval around $\frac{\gamma}{5}$
        and is hardly ever be larger than $\gamma$ for large $n$
        according to the law of large numbers. While this is only the
        intuition, a similar reasoning in \cite{svitkina2008submodular} offers
        more details.
	
        Therefore, the output $\hat{X}$ of any polynomial algorithm
        must satisfies $g(\hat{X})\leq \max_{X\leq k}g'(X)$ since,
        otherwise the algorithm actually distinguishes the two
        function at $\hat{X}$,
        $g(\hat{X})>\max_{X\leq k}g'(X)\geq g'(\hat{X})$. The
        approximate ratio
        $\frac{g(\hat{X})}{\text{OPT}}\leq
        \frac{\text{OPT}'}{\text{OPT}}=\frac{\alpha-\curv^g
          \alpha}{\alpha-\curv^g \gamma}=(1-\curv^g)\frac{1}{1-\curv^g
          \sqrt{\frac{\omega(\log n)}{n}}}\leq
        1-\curv^g+\epsilon$. Therefore, the guarantee of any
        polynomial algorithm, that, by definition, the best case
        approximate ratio, is no greater than $1-\curv^g+\epsilon$ for
        any $\epsilon>0$ since, otherwise contradicts the information
        theoretic hardness.

\end{proof}

\section{Proof of Theorem~\ref{theo:hardnessPmatroid}}
\label{sec:hardnessPmatroid}
\hardnessPmatroid*
\begin{proof}
	Consider the $p$-set problem~\cite{hazan2006complexity}, let $R$ be the maximum disjoint sets of
	these $p$ sets.  No polynomial algorithm can find a larger number of disjoint sets
	than $O(\frac{\ln p}{p})|R|$~\cite{hazan2006complexity}. Let $k =O(\frac{\ln p}{p})|R|$.  So
	no polynomial algorithm can find a feasible set with size larger than $k$
	unless P=NP.
	
	Let $h(X)=(1-\beta)|X|+\beta\max\{|X|-k,0\}$. It is easy to check
	that $h$ is a BP function with $f=0$ and $g = h$ with
	$\curv^g=\beta$.
	
	Therefore, the output $\hat{X}$ of any polynomial algorithm
        that maximizes $h$ under the $p$-set constraint (expressible via the intersection
        of $p$ matroids) must satisfy that $|X|\leq k$ and, therefore,
        $h(\hat{X})\leq (1-\beta)k$ unless P=NP. But
        $h(X^*)\geq h(R)=(1-\beta)|R|+\beta (|R|-k)= |R| -\beta k$.
	
	Thus, the approximate ratio
        \begin{align}
          \frac{h(\hat{X})}{h(X^*)}
          &\leq \frac{(1-\beta)k}{|R|-\beta k}
          \leq 	\frac{(1-\beta)O(\frac{\ln p}{p})}{1-\beta O(\frac{\ln p}{p})}
          \leq \frac{(1-\beta)O(\frac{\ln p}{p})}{\frac{1}{2}}=(1-\curv^g)O(\frac{\ln p}{p}).
        \end{align}
	since the denominator $1-\beta O(\frac{\ln p}{p})\geq \frac{1}{2}$ asymptotically and $2O(\frac{\ln p}{p})=O(\frac{\ln p}{p})$.
\end{proof}

\section{Submodularity Ratio and Generalized Curvature}
\label{sec:comp-subm-ratio}

In this section, we compare the pair $\curvf,\curvg$ of curvatures
with the submodularity ratio
\cite{das2011submodular,2017arXiv170302100B}.  We also show that both
the generalized curvature introduced in \cite{2017arXiv170302100B} and
the submodularity ratio \cite{das2011submodular} appears to be hard to
compute in general under the oracle model.  Lastly, we compare the
pair $\curvf,\curvg$ with another notion of curvature introduced in
\cite{DBLP:journals/corr/SviridenkoW13}, showing a simple inequality
relationship in general and a correspondence when $h=g$.

\subsection{Submodularity ratio}

The submodularity ratio is defined as
\begin{align}
\gamma_{U,k}(h) = \min_{L\subseteq U, S:|S|\leq k, S\cap
  L=\emptyset}\frac{\sum_{x\in S}h(x|L)}{h(S|L)}
\label{eqn:submod_ratio_expanded}
\end{align}
with $U \subseteq V$ and $1 \leq k \leq |V|=n$, and typically we
consider $\gamma_{V,n}$.  We can establish a simple lower bound of the
submodularity ratio based on the supermodular curvature as follows.
\begin{lemma}
$\gamma_{V,n}(h)\geq 1-\curvg$ when $h=f+g$.
\end{lemma}
\begin{proof}
  For all $L\subseteq V$ and $S\cap L=\emptyset$, we have $\frac{\sum_{x\in S}h(x|L)}{h(S|L)}\geq 1-\curvg$
  which follows from
  Lemma~\ref{lemma:curvInequ}\ref{lemma:curvInequ4}
  Thus, $\gamma_{V,n}(h)\geq 1-\curvg$.
\end{proof}
The function $h$ is submodular if and only if $\gamma_{V,n} = 1$ so
one might hope that given a BP function $h=f+g$, that as
$\gamma_{V,n}(h) \to 1$, correspondingly $\curvg \to 0$. This is not the
case, however, as can be seen by considering the following example.

Let $a$ be an element of $V$ and define the function
$g(A)=|A\cap(V\setminus\set{a})|+\epsilon|A\cap(V\setminus\set{a})||A\cap\set{a}|$,
where $\epsilon>0$ is a very small number. Immediately, we have that
$g$ being supermodular and monotone. Also note, if $a\notin A$ then
$g(A)=|A|$; if $a\in A$ then $g(A)=(|A|-1)(1+\epsilon)$.

First, we calculate the supermodular curvature $\curvg$. We have that
$g(a)=0$ and also $g(a|V\setminus\set{a})=\epsilon(n-1)$. Therefore,
the function is fully curved, $\curvg=1$.

Next, we calculate the submodularity ratio
$\gamma_{V,n}= \min_{L,S\subset V, S\cap L=\emptyset}\frac{\sum_{v\in
    S}g(v|L)}{g(S|L)}$.  When $|S| = 1$,
$\frac{\sum_{v\in S}g(v|L)}{g(S|L)}=1$.  When $|S|\geq 2$, we have 
the following 3 cases (recall that $S \cap L = \emptyset$ so there is no forth case):
\begin{itemize}
	\item $a\in S$. $g(S|L)=g(S\cup L)-g(L)= (|S|+|L|-1)(1+\epsilon) -|L|$ is very close to $|S|-1$ for very small $\epsilon$. $\sum_{v\in S}g(v|L)=\epsilon|L|+|S|-1$, which is also very close to $|S|-1$ for small $\epsilon$. So $\frac{\sum_{v\in S}g(v|L)}{g(S|L)}\approx 1$ for small $\epsilon$.
	\item $a\in L$. $g(S|L)=g(S\cup L)-g(L) = |S|(1+\epsilon)$. $\sum_{v\in S}g(v|L)=|S|(1+\epsilon)$. So $\frac{\sum_{v\in S}g(v|L)}{g(S|L)}= 1$
	\item $a\notin S\cup L$. $g(S|L) = |S|$ and $\sum_{v\in S}g(v|L) = |S|$. Therefore, $\frac{\sum_{v\in S}g(v|L)}{g(S|L)}= 1$.
\end{itemize}

In all cases, $\frac{\sum_{v\in S}g(v|L)}{g(S|L)}$ is either 1 or very
close to 1 for small $\epsilon$, so $\gamma_{V,n}$ has only 1 as an
upper bound. That is, we have an example function that is purely
supermodular and fully curved ($\curvg=1$) for all non-zero values of
$\epsilon$, but the submodularity ratio can be arbitrarily close to 1.
If we consider a weighted sum of a submodular function and this
supermodular function, the submodularity ratio is again arbitrarily
close to 1. Therefore, there does not seem to be an immediately
accessible strong relationship between the supermodular curvature and the
submodularity ratio.

\subsection{Hardness of Generalized Curvature and Submodularity Ratio}
\label{sec:hardn-gener-curv}

The generalized curvature \citet{2017arXiv170302100B} of a
non-negative function $h$ is the smallest scalar $\alpha$ s.t.\ 
\begin{align}
h(v|S\setminus\set{v}\cup\Omega) \geq (1-\alpha) h(v|S\setminus\set{v})
\end{align}
for all $S,\Omega\subseteq V$ and $v\in S\setminus \Omega$ and this is
used, in concert with the submodularity ratio, to produce bounds such
as $\frac{1}{\alpha}(1-e^{-\alpha \gamma})$ for the greedy algorithm.
Unfortunately, the generalized
curvature is hard to compute under the oracle model.
We have the following.

\begin{lemma}
  There exists an instance of a non-negative function $h$ whose generalized curvature can not be calculated
  in polynomial time, when we have only oracle access to the function.
\end{lemma}
\begin{proof}
	We consider a non-negative function $h':2^V\rightarrow R$ with ground set size equals $n$ ($n$ is even number). Let $h'(X)=|X|$ for all $X\subseteq V$. Let $R\subseteq V$ be an arbitrary set
	with $|R| = \frac{n}{2}$. Define another set function $h:2^V\rightarrow R$, $h(X)=h'(X)$ for all $X\subseteq V$ and $X\neq R$; $h(R)=\frac{n}{2}-1$.
	
	First, we can easily calculate the generalized curvature of
        $h'$ and $h$. We have that $\alpha_{h'}=0$ since $h'$ is a
        non-decreasing modular function. For $h$, let
        $S\cup \Omega = R$, $S\cap \Omega = \emptyset$,
        $|S|,|\Omega|\geq 1$ and $v\in S$, we have
        $h(v|S\setminus\set{v}\cup\Omega) = 0$ and
        $h(v|S\setminus\set{v}) = 1$. Therefore $\alpha=1$ is the
        smallest scalar s.t.
        $h(v|S\setminus\set{v}\cup\Omega) \geq (1-\alpha)
        h(v|S\setminus\set{v})$. So, as a conclusion of this part, the
        generalized curvature of the two functions are not the same.
	
	 Next we use a proof technique similar
	 to~\cite{svitkina2008submodular}. Note that $h'(X)= h(X)$ if and only
	 if $X\neq R$. So for any algorithm trying to calculate $\alpha_h$, before it
	 evaluates $h(R)$, all function evaluations are the same with
	 calculating $\alpha_{h'}$.  Additionally, since $h(X)=|X|$, it
	 is permutation symmetric. Therefore, the algorithm can only do
	 random search to find $R$. If the algorithm acquires a polynomial
	 number $O(n^m)$ of sets of size $\frac{n}{2}$, the probability of finding
	 $R$ is
	 $\frac{O(n^m)}{\binom{n}{\frac{n}{2}}}\leq \frac{O(n^m)}{(n/{\frac{n}{2}})^{\frac{n}{2}}} =
	 \frac{O(n^m)}{2^{n/2}}\leq O(2^{-n/2+\epsilon n})$ for all
	 $\epsilon>0$. 
	 
	 Therefore, no algorithm can be guaranteed to distinguish $h$ and $h'$ in polynomial time. Since the generalized curvature of $h$ and $h'$ are different, neither of them can be calculated in polynomial time.
	
	\end{proof}

        Likewise, the submodularity ratio is unfortunately also hard
        to compute exactly, in the oracle model.
        
\begin{lemma}
  There exists an instance of a non-negative function $h$ whose
  submodularity ratio (Equation~\eqref{eqn:submod_ratio_expanded}) can
  not be calculated in polynomial time under only oracle access 
  to that function.
\end{lemma}
\begin{proof}
  We consider a non-negative function $h':2^V\rightarrow R$ with
  ground set size $n$ (where $n$ is an even number). Let $h'(X)=|X|$ for
  all $X\subseteq V$. Let $R\subseteq V$ be an arbitrary set with
  $|R| = \frac{n}{2}$. Define another set function
  $h:2^V\rightarrow R$, $h(X)=h'(X)$ for all $X\subseteq V$ and
  $X\neq R$ and $h(R)=\frac{n}{2}-1$.
	
  We can easily calculate the submodularity ratio of both $h'$ and $h$
  as follows. We have that $\gamma_{V,n}(h')=1$ since $h'$ is a
  non-decreasing modular (and thus submodular) function. For $h$,
  choose an element $v_1\in R$ and another element
  $v_2 \in V\setminus R$, and let $L=R\setminus \set{v_1}$ and
  $S=\set{v_1,v_2}$. We have
  $\frac{\sum_{v\in
      S}h(v|L)}{h(S|L)}=\frac{h(R)+h(R\setminus\set{v_1}\cup\set{v_2})-2h(R\setminus\set{v_1})}{h(R\cup\set{v_2})-h(R\setminus\set{v_1})}=\frac{1}{2}$ and thus
  $\gamma_{V,n}(h) = \min_{L,S\subseteq V, S\cap
    L=\emptyset}\frac{\sum_{v\in S}h(v|L)}{h(S|L)}\leq \frac{1}{2}$.
  Therefore, the submodularity ratio of the two functions are not the
  same. Given the submodularity ratio of the two functions, we would
  be able to tell them apart.
	
  Next we use a proof technique similar
  to~\cite{svitkina2008submodular}. We have that $h'(X)= h(X)$ if and
  only if $X\neq R$. So for any algorithm trying to calculate
  $\gamma_{V,n}(h)$, before it evaluates $h(R)$, all function evaluations are
  the same with calculating $\gamma_{V,n}(h')$.  Additionally, since
  $h(X)=|X|$ is permutation symmetric, the algorithm
  can only do a random search to find $R$. If the algorithm queries a
  polynomial number $O(n^m)$ of sets of size $\frac{n}{2}$, the
  probability of finding $R$ is
  $\frac{O(n^m)}{\binom{n}{\frac{n}{2}}}\leq
  \frac{O(n^m)}{(n/{\frac{n}{2}})^{\frac{n}{2}}} =
  \frac{O(n^m)}{2^{n/2}}\leq O(2^{-n/2+\epsilon n})$ for all
  $\epsilon>0$.

  Therefore, no algorithm can guarantee to distinguish $h$ and $h'$ in
  polynomial time. Since the submodularity ratio of $h$ and $h'$ are
  different, this means that neither of them can be calculated in
  polynomial time.
	
\end{proof}

\subsection{Comparison to \citet{DBLP:journals/corr/SviridenkoW13}'s curvature}
\label{sec:comp-citetdblp:j-cur}

\citet{DBLP:journals/corr/SviridenkoW13} (in their Section 8) define a notion of curvature 
as follows:
\begin{align}
1-c=\min_j \min_{A,B\subseteq V\setminus j} \frac{h(j|A)}{h(j|B)}
\label{eqn:sviridenkow_curv}
\end{align}
We can establish a simple upper bound on $c$ based on submodular and
supermodular curvature as follows. We calculate
$\frac{h(j|A)}{h(j|B)}$ given $h=f+g$ and $\kappa_f$ and $\kappa^g$ as
follows.  First, $f(j|B)\leq f(j) \leq \frac{1}{1-\kappa_f}f(j|A)$
which follows from Lemma~\ref{lemma:curvInequ}~\ref{lemma:curvInequ1}.
Thus $\frac{f(j|A)}{f(j|B)} \geq 1-\kappa_f$.  Next,
$g(j|A)\geq g(j) \geq (1-\kappa^g)g(j|B)$ which follows from
Lemma~\ref{lemma:curvInequ}~\ref{lemma:curvInequ2}.  Thus,
$\frac{g(j|A)}{g(j|B)} \geq 1-\kappa^g$.  Therefore,
\begin{align}
	\frac{h(j|A)}{h(j|B)}&=\frac{f(j|A)+g(j|A)}{f(j|B)+g(j|B)}\geq \frac{(1-\kappa_f)f(j|B)+(1-\kappa^g)g(j|B)}{f(j|B)+g(j|B)} \\ &\geq \frac{\min(1-\kappa_f,1-\kappa^g)(f(j|B)+g(j|B))}{f(j|B)+g(j|B)} \geq \min(1-\kappa_f,1-\kappa^g)
\end{align}
Thus we have $1-c\geq \min(1-\kappa_f,1-\kappa^g)$, or
$c\leq \max(\kappa_f,\kappa^g)$.

Note that for purely supermodular functions, $\kappa_f = 0$ and,
considering Equation~\eqref{eqn:sviridenkow_curv}, we have
$c = \kappa^g$. This coincides with the $1-\curvg$ bound and hardness
for monotone supermodular functions --- compare Theorem 8.1 of
\citet{DBLP:journals/corr/SviridenkoW13} with the present paper's
item~\ref{item:one_minus_curvg_bound_card} in
Section~\ref{sec:card-constr} and Theorem~\ref{theo:hardnessCard}.

\end{document}